\documentclass{pracalicmgr}
\usepackage{polski}
\usepackage[english]{babel}
\usepackage[utf8]{inputenc}
\usepackage{physics}
\usepackage{braket} 

\usepackage{graphicx}
\usepackage{booktabs}
\usepackage{float}

\usepackage{amsmath,mathtools}
\usepackage{amssymb}
\usepackage{tikz}
\usepackage{wrapfig}
\usepackage{tabularx}
\usepackage{ bbold }
\usepackage[font=small]{caption}
\usepackage{listings}

\definecolor{blue(pigment)}{rgb}{0.2, 0.2, 0.6}
\definecolor{coolblack}{rgb}{0.0, 0.18, 0.39}
\definecolor{darkblue}{rgb}{0.0, 0.0, 0.55}

\usepackage[noframe]{showframe} 
\usepackage{afterpage} 
\usepackage{svg}
\usepackage{breqn} 

\newcommand{\I}{i} 
\newcommand{\nth}{\textsuperscript{th}}
\newcommand{\dgr[1]}{#1^{\circ}}
\newcommand{\Csq}{$\left| \mathcal{C} \right|^2~$}
\newcommand{\Cesq}{$\left| \mathcal{C}_{\vec{e}} \right|^2~$}
\newcommand{\refsec[1]}{\textbf{\ref{#1}}}
\newcommand{\newterm[1]}{\emph{#1}}

\renewcommand{\Tr}{\text{Tr}\!}
\newcommand{\identity}{\hat{\mathbb{1}}}
\newcommand{\vac}{\ket{0}_{\text{vac}}}

\usepackage{cite}
\usepackage{hyperref}
\usepackage{amsthm} 
\newtheorem{theorem}{Theorem}
\newtheorem{lemma}{Lemma}
\usepackage{xfrac} 

\author{Maciej Kościelski}

\nralbumu{404542}

\title{Dynamical Bose-Hubbard model for entanglement generation and storing}

\tytulpl{Dynamiczny model Bosego-Hubbarda do wytwarzania i przechowywania stanów splątanych}

\kierunek{Fizyka}

\specjalnosc{Fizyka Teoretyczna}

\opiekun{\\dr hab. Emilia Witkowska, prof. IP PAS\\ Division of Theoretical Physics\\ Institute of Physics PAS\\ and\\ dr hab. Jan Chwedeńczuk\\ Chair of Quantum Optics and Atomic Physics\\ Institute of Theoretical Physics, UW}

\dziedzina{13.2 Fizyka}

\date{April 2021}


\grantnote{Praca magisterska napisana w Instytucie Fizyki PAN w Warszawie, na podstawie badań prowadzonych w ramach projektu finansowanego przez NCN z konkursu SONATA BIS-5 Nr UMO-2015/18/ST2/00760}

\begin{document}

	\maketitle
	\let\cleardoublepage\clearpage

	\begin{abstract} 

		This work presents a theoretical study of a protocol for dynamical generation and storage of the durable, highly entangled Greenberger-Horne-Zeilinger (GHZ) state in a system composed of bosonic atoms loaded into a one-dimensional optical lattice potential. A method of indicating entanglement in the system is also presented.
		
		The system ground-state can be either in the superfluid phase or in the Mott insulator phase. The wave functions of atoms in the superfluid phase are delocalised over the whole lattice and overlap. In the Mott phase, the wave functions are localised around lattice sites. The GHZ state is being generated in the superfluid phase and stored in the Mott insulator phase. It is achieved by a linear change of the potential depth in an optical lattice filled with atoms of two species. The numerical method used to describe the system is based on the exact diagonalisation of the Bose-Hubbard Hamiltonian. A quantum correlator indicating the level of multi-mode entanglement is introduced. Finally, it is shown that the value of the correlator indicates generation of the GHZ state. The appearance of the GHZ state is confirmed by the numerical representation of the state.

		\vspace*{10pt}
		
		\begin{center}
			\bfseries\large Streszczenie
		\end{center}

		Niniejsza praca przedstawia teoretyczną analizę protokołu do wytwarzania i przechowywania trwałego, maksymalnie splątanego stanu Greenbergera-Horne'a-Zeilingera (GHZ) w układzie złożonym z bozonów w jednowymiarowej sieci optycznej. Przedstawiona jest również metoda wykrywania splątania w układzie.

		Stan podstawowy układu może występować w fazie nadciekłej albo w fazie izolatora Motta. W przypadku fazy nadciekłej funkcje falowe opisujące atomy są zdelokalizowane w ramach całej sieci i mogą się przekrywać. W fazie Motta funkcje falowe są zlokalizowane wokół oczek sieci. Stan GHZ jest wytwarzany w fazie nadciekłej, a przechowywany w fazie izolatora Motta. Jest to osiągnięte poprzez liniową zmianę głębokości potencjału pułapkującego sieci optycznej wypełnionej bozonowymi cząstkami dwóch rodzajów. Użyta metoda numeryczna opiera się na ścisłej diagonalizacji hamiltonianu Bosego-Hubbarda. Wprowadzony zostaje korelator kwantowy służący do pomiaru poziomu splątania. Wartość korelatora wskazuje na występowanie stanu GHZ. Jego wytworzenie potwierdza również analiza numerycznej reprezentacji uzyskanego stanu.

	\end{abstract}

	\section*{Acknowledgements}
    \thispagestyle{empty}

    I would like to express my gratitude to the people who have supported me throughout my academic endeavours.
    
    My supervisors: dr hab. Emilia Witkowska, prof. IP PAS and dr hab. Jan Chwedeńczuk -- thank you for the initial idea of this work, the valuable suggestions on the thesis, and everything I have learnt from you.
    
    Everyone who gave me a consultation on writing and proofreading in English -- thank you for your scrupulosity and patience.
    
    I would also like to extend my sincere thanks to everyone who has supported me emotionally during the writing of the thesis.

    \clearpage

	\tableofcontents\thispagestyle{empty}
	
	\clearpage\pagenumbering{arabic}

	\chapter{Motivation and relation to other work}
		
	The aim of this work is to present a method for generating and storing an entangled state in a system composed of ultracold atoms loaded into an optical lattice. The considered system and protocol can be realised experimentally with nowadays techniques \cite{Gross995, RevModPhys.80.885}. The method can be exploit for experimental generation of entanglement in the many-body system. It can also be used as a framework for theoretical studies of multi-mode entanglement.
	
	The main concept underlying this work is entanglement \cite{entanglement_review}. Discussion of this phenomenon starts with the groundbreaking paper by Albert Einstein, Boris Podolsky and Nathan Rosen from 1935 \cite{EPR}. In that paper, the authors showed that in quantum theory a measurement conducted on one particle can give a certain information about the result of some measurement on another particle. They proven that this is independent of the interval separating the particles, what implies a non-local correlation between observables. In the same year, another paper was published by Erwin Schr\"odinger where entanglement has been discussed as a general phenomenon of quantum theory \cite{entanglement_schr}.

	The phenomenon of entanglement can be understood in the following way: A group of quantum-mechanical objects are in entangled state if one can not describe behaviour of each of them separately, but rather has to consider the state of the whole system in order not to lose any information about the behaviour of its parts. Entanglement has already been proven to be of great importance in many areas like quantum information \cite{QI_book, QI_rev}, quantum cryptography \cite{cryptography} and other branches of quantum theory. Another important field that exploits entanglement is quantum metrology, the idea of which is to perform precise measurements of physical quantities, for example in quantum interferometry \cite{interfer1, interfer2, interfer3, interferReview, ent_and_interfer}.
	Entanglement is a distinctive feature of quantum theory. It gives plenty of possibilities to experimentally test the advantage of quantum theory over the classical theory in describing physical phenomena \cite{quantum_tests1, quantum_tests2}.
	The entangled states can be produced using many techniques: in quantum optics (with photons \cite{photonic_QI}), solid state physics \cite{solid-state_qubits, quantum_dot} or cold matter physics (using ions or neutral gases of ultracold atoms \cite{ions, ultracold_gases}).

	One of the most extensively exploited entangled states is the Greenberger-Horne-Zeilinger (GHZ) state, which is the maximally entangled state. It was first proposed in \cite{GHZ_paper} to test the fundamental issue of non-locality in quantum theory. For a system of $M$ qubits (each qubit from orthonormal basis $\left\{ \ket{0}, \ket{1} \right\}$) the GHZ state is defined as
	\begin{equation} \label{GHZ}
	\ket{\text{GHZ}} := \frac{\ket{0}^{\otimes M}+\ket{1}^{\otimes M}}{\sqrt{2}}.
	\end{equation}
	This state has been found useful in quantum information theory \cite{Gao_2004, Xia_2006}. As a superposition of two macroscopic states it is fragile to decoherence, and thus, difficult to be produced in laboratory, in particular for a large number of qubits. The first successful attempt to produce it was made in 1999 with use of three photons \cite{GHZ-produced}. The state-of-the-art experiments achieve a GHZ state composed of up to 27 qubits \cite{GHZ27}. This work, following \cite{paper}, proposes a protocol allowing to create the GHZ state using an arbitrary number of bosonic atoms in optical lattice.
	
	This work focuses on a system composed of bosonic atoms loaded into an optical lattice potential. The ultracold regime is considered so that the Bose-Hubbard model can be used as an approximation giving sufficient description of the system \cite{BH_model, phaseDiagram_BH}. The system can be observed in two distinct thermodynamic phases \cite{phaseDiagram_BH, Mott_SF_in_BH} -- the superfluid phase characterised by the presence of Bose-Einstein condensate \cite{BEC} and the Mott insulator phase characterised by the fact that all particles are well-localised at the lattice sites.
	
	This thesis is motivated by the results of two papers \cite{paper, Chwe20}. In \cite{paper} the authors propose a protocol for generation and storage spin-squeezed states in the Mott phase. Additionally, they show that the protocol allows one to generate the GHZ state and store it in the optical lattice. The idea of this thesis is to characterise and prove that indeed the highly entangled state can be generated and stored in the Mott insulator phase using the above method. To quantify the level of entanglement in this thesis a generalised notion of the correlator defined in \cite{Chwe20} is used. The generalisation bases on the theorems \refsec[the:Theorem1] and \refsec[the:MaxC], which are proved in appendices. The value of the correlator is bounded from above for separable (non-entangled) states. Breaking the bound proves that the state is entangled \cite{Chwe20}.
	
	It was also shown that the correlator quantifies the level of entanglement in the system \cite{Chwe20}.

	\chapter{Description of the problem} 

	In quantum statistical physics, there are two classes of particles -- fermions and bosons. Fermionic particles (\newterm[fermions]) are characterised by the property that any wave function describing a set of them is anti-symmetric due to exchange of the particles. This leads to an important fact that only one fermion can occupy a state described by specific quantum number. The other type of particles -- \newterm[bosons] -- is characterised by the symmetry of the wave function due to exchange of particles. This gives a possibility of arbitrary many bosons taking the exact same quantum state.

	This work examines bosonic particles at ultra-low temperatures of the order of hundreds of nanokelvins. A gas composed of bosons at such low temperatures exhibits a phase transition to the so-called Bose-Einstein condensate. This phenomenon was theoretically predicted in 1924 by Satyendra Nath Bose \cite{bose1924} and Albert Einstein \cite{einstein1924}. It was experimentally realised for the first time in 1995 by Eric Cornell, Carl Wieman et al. \cite{BECanderson1995} and, independently, in the same year by Wolfgang Ketterle et al. \cite{BECketterle1995}. In 2001, the experiments were honoured with the Nobel prize in physics \cite{Nobel2001}.
	
	In a bosonic gas below the critical temperature $T_c$ (usually $T_c \sim 100$~nK for atoms) a macroscopic number of particles tends to occupy the lowest energy state. The macroscopic quantum state composed of the bosons occupying the ground state is called the \newterm[Bose-Einstein condensate] (BEC). The amount of BEC in a gas can be expressed by a notion of the \newterm[condensate fraction]
	\begin{equation}
	    f_c := \frac{\braket{n_0}}{N},
	\end{equation}
	where $\braket{n_0}$ is the average number of particles in the lowest energy state, and $N$ is the total number of particles in the gas. A gas with $f_c\approx1$ is called pure condensate.

	All the particles forming BEC are indistinguishable. It allows to apply a quantum operation to a large set of particles simultaneously. For this reason, using cold gases, one has more potentiality in manipulating many-body quantum systems than using solid states. Thanks to those properties, a whole branch of research projects exploiting BEC has emerged. They include investigations of interference of two condensates \cite{andrews1997interference, anderson1998interfer}, applications of BEC in quantum metrology \cite{interferReview, cronin2009optics, wasak2018bell, BECclock}, producing solitons \cite{solitons, perez1998bose, khaykovich2002soliton}, Feshbach resonances \cite{feshbach} and squeezed or entangled states \cite{interferReview, paper, orzel2001squeezed, hamley2012spin} in Bose-Einstein condensate. Experiments on the generation of BEC are being conducted all over the world including KL FAMO laboratory in Toruń \cite{FAMO} and even in space \cite{inspaceHNV, ISS}.

	This research investigates an ultracold gas of bosonic atoms, which may exhibit a Bose-Einstein condensation. It will be referred to when discussing the condensate fraction in the further sections. It is, however, worth mentioning that BEC is not the key aspect of the study. This work examines bosonic atoms loaded into an optical lattice, which might have properties of BEC in some regime of parameters what will be discussed in Section~\refsec[sec:phaseDiagram]. In other ranges of parameters, the system is in the Mott insulator phase. 
	In the end it is worth noting that this study investigates a finite number of particles; thus, referring to thermodynamic phenomena such as BEC cannot be considered in a strict sense.

	\section{Optical lattice} 

	The protocol proposed in this thesis relies on optical lattice -- a device, which, using lasers, creates an effective periodic trapping potential for ultracold atoms.

	An optical lattice is made of counter-propagating laser beams. To keep the atoms in the optical lattice potential, they are first cooled down to ultra-low temperatures of the order of hundreds of nanokelvins using various methods of cooling and trapping \cite{cooling_review, cooling&trapping}. The whole setup is closed in a vacuum chamber to isolate the system from the environment to avoid undesirable effects, e.g. scattering with other particles. The atomic gas is also very dilute to prevent many-body scattering events. Those events would lead to emergence of molecules, which would not be trapped by the optical lattice leading to vaporisation of the atoms in a short time.

    \begin{quote}
        \subsection{Potential of the lattice}

    	The counter-propagating laser beams are tuned in such a way that the electro-magnetic field forms a standing wave. Following \cite{Griffiths}, in the one-dimensional case, the electric field $\mathbf{E}$ of the standing wave can be written as
    	\begin{equation}
    	    \mathbf{E}(x,t) = \mathbf{\hat{e}}_\perp \tilde{E}(x) \sin{(\omega_L t)} := \mathbf{\hat{e}}_\perp E_0 \sin{(k_L x)} \sin{(\omega_L t)},
    	\end{equation}
    	with $\mathbf{\hat{e}}_\perp$ -- the unit \newterm[polarisation vector] perpendicular to the $x$-axis, $\tilde{E}(x)$ -- \newterm[field amplitude] in each point $x$ of space, $E_0$ -- amplitude of the electro-magnetic field, $k_L=\sfrac{2 \pi}{\lambda_L}$ \newterm[wavevector] of the monochromatic light of \newterm[wavelength] $\lambda_L$, $\omega_L=\sfrac{2 \pi c}{\lambda_L}$ -- \newterm[angular frequency] of the light and $t$ -- time.
    	
    	If a neutral atom is surrounded by such a field it induces an atomic dipole moment $\mathbf{p}$:
    	\begin{equation}
    	    \mathbf{p}(x,t) = \mathbf{\hat{e}}_\perp \alpha \tilde{E}(x) \sin{(\omega_L t)},
    	\end{equation}
    	where $\alpha$ being called \newterm[complex polarisability] is dependent on the frequency $\omega_L$ of the light \cite{dipole}. Interaction of the atom with the field creates \newterm[dipole potential] $V$:
    	\begin{equation*}
    	    V(x) = - \frac{1}{2} \int_0^{\frac{2 \pi}{\omega_L}} \mathbf{p} \cdot \mathbf{E} ~dt = - \text{Re}(\alpha) \left| \tilde{E} \right|^2 =
    	\end{equation*}
    	\begin{equation}
    	    = - \text{Re}(\alpha) E_0^2 \sin^2\left(\frac{2 \pi}{\lambda_L} x\right),
    	\end{equation}
    	where integral denotes time-averaging over one period of the optical wave \cite{dipole}.
    	
    	It is apparent that two counter-propagating laser beams of wavelength $\lambda_L$ can create a potential with \newterm[spatial period] $d:=\sfrac{\lambda_L}{2}$ for a neutral atom. This periodic structure of the potential with ability to trap atoms is called \newterm[optical lattice]. The following Figure~\ref{fig:wave} presents a sketch of the one-dimensional optical lattice. The red arrows illustrate laser beams while the sinusoidal line the potential created by the standing electro-magnetic wave.
    	\begin{figure}[H]
    	    \centering
    	    \includegraphics[width=0.8\textwidth]{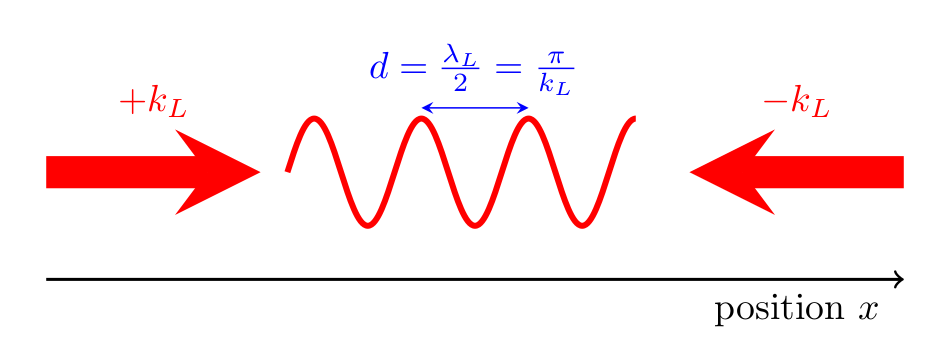}
    	    \caption{A scheme of the one-dimensional optical lattice potential. The red thick arrows symbolise two counter-propagating laser beams of the same wavelength $\lambda_L$ and the opposite wavevectors $\pm k_L$. The standing wave generated by the two laser beams creates a periodic potential for atoms illustrated by the sinusoidal line. The minima of the potential, where the ultracold atoms tend to congregate, are separated by the distance $d=\lambda_L/2$. The minima are called \newterm[lattice sites]. Adapted from \cite{SuperfluidFraction}.}
    	    \label{fig:wave}
    	\end{figure}
	\end{quote}

	What is highly useful is that one can precisely control parameters of the lattice such as geometry of the lattice and its depth during the experiment by changing the frequency of the laser beams, their intensity and the angle between them. Using two counter-propagating laser beams one gets one-dimensional periodic structure of the lattice potential. One can create two- or three-dimensional structures with four or six laser beams (see Figure~\ref{latticeScheme}). The elementary cell of the lattice is cubic when the laser beams are perpendicular. A system of triclinic symmetry can be also created when changing the angles between the laser beams. In the whole lattice there are thousands of atoms, usually few per one site. However, one can very precisely control the number of atoms in each site. The opportunity of changing so many parameters makes optical lattices a basic element of many experiments like the most precise device available to humans -- the optical lattice clock \cite{mcgrew2018atomic, nicholson2015systematic, RevModPhys.83.331, RevModPhys.87.637}.

    \begin{figure}[H]
		\centering
		\includegraphics[width=0.7\textwidth]{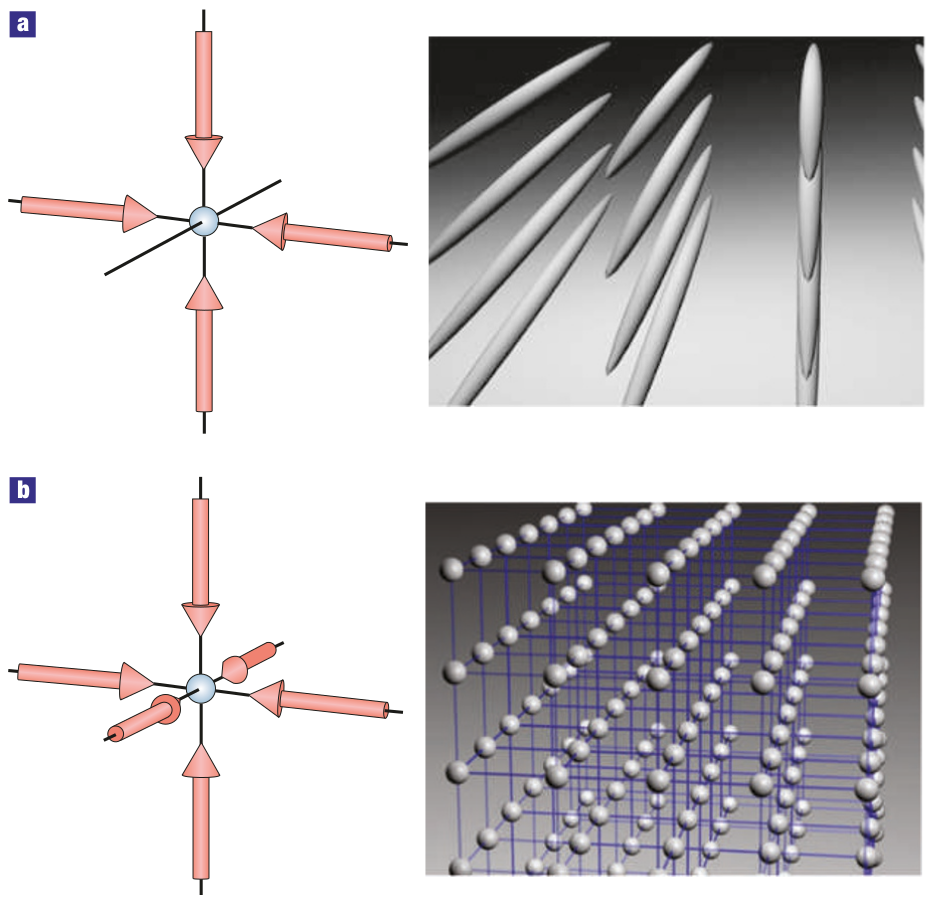}
		\caption{The upper figure (a) symbolises 2-dimensional optical lattice. The left part shows two perpendicular pairs of counter-propagating laser beams (red arrows) shining on a group of atoms (the small ball inside). The left part of the figure presents 2-dimensional periodical structure made of clouds of ultracold atoms generated by the geometry of electro-magnetic field. In the lower figure (b) three perpendicular laser beams (left side) generate 3-dimensional periodicity (right side) resembling a crystal. Reprinted from \cite{opt_latt_Bloch}.}
		\label{latticeScheme}
	\end{figure}

    \subsection{Parameters of the lattice}

	This work models bosonic atoms loaded into the one-dimensional optical lattice potential
	\begin{equation}
	\label{potentialV}
	V(x) = - \text{Re}(\alpha) E_0^2 \sin^2 \left( \frac{2 \pi}{\lambda_L} x \right) =: V_0 \sin^2 \left(  k_L ~x \right),
	\end{equation}
	where $V_{0}:=- \text{Re}(\alpha) E_0^2$ is called \newterm[lattice depth] or \newterm[lattice height], $k_L = 2\pi/\lambda_L$ is the wave number of trapping light, and $d=\lambda_L/2$ is the spatial \newterm[period of the lattice]. For typically used lasers, $\lambda_L\approx800$~nm.

	Using $k_L$ and mass $m$ of a single atom, one can introduce a natural energy scale,
	\begin{equation}
	E_R := \frac{k_L^2 \hbar^2}{2 m} = \frac{ h^2}{2 m \lambda_L^2}.
	\end{equation}
	In further sections the \newterm[recoil energy] $E_R$ is used as the energy unit. This value is typically close to $2.4 \cdot 10^{-30} \text{J}$.

	In the numerical calculations, dimensionless variables like
	\begin{equation}
	\tilde{V}_0 := \frac{V_0}{E_R}
	\end{equation}
	are used. The value of $\tilde{V}_0$ in this study is usually in the range $\tilde{V}_0\in[3, 45]$.

	This thesis examines $N=6$ atoms on $M=6$ lattice sites, it also assumes periodic boundary conditions. The system can be thus modelled using the Bose-Hubbard model \cite{opt_latt_Bloch, ions} introduced in the next paragraph.

	\section{The Bose-Hubbard model} \label{sec:BoseHubbardModel}
	
	The dynamics of the optical lattice filled with bosons can be approximately described by the so-called \newterm[Bose-Hubbard model] \cite{ions}. The approximation assumes that the particles are localised near the lattice sites, the tunnelling of the particles between the lattice sites can only occur between the nearest-neighbour sites and that the only possible states are those in the lowest energy band.

	In this work an ultracold gas composed of a few bosonic atoms with repulsive interaction in two internal states $A$ and $B$, loaded into a three-dimensional potential of an optical lattice is considered. It is assumed that tunnelling is only possible along the $x$ direction while transversely the atoms are described by a localised wave function that can be approximated by a Gaussian function with a characteristic length $L_\perp$. The optical lattice potential in the $x$-direction is $V(x)$ \eqref{potentialV}. In the absence of the lattice, the system is homogeneous and the atoms are confined in a \newterm[flat-bottom potential] \cite{PhysRevLett.110.200406}. After reduction of the perpendicular directions, the effective one-dimensional system Hamiltonian reads
	\begin{eqnarray}
		\hat{\mathcal{H}} &=& \sum_{\sigma=A,B} \int dx \left(
		\hat{\Psi}_\sigma^\dagger(x) \hat{h} \hat{\Psi}_\sigma(x) + 
		\frac{g_{\sigma\sigma}}{2} 
		\hat{\Psi}_\sigma^{2\dagger}(x) \hat{\Psi}_\sigma^{2}(x) \right) \nonumber \\
		&+& g_{AB} \int dx \: \hat{\Psi}_A^\dagger(x)\hat{\Psi}_B^\dagger(x)\hat{\Psi}_B(x)\hat{\Psi}_A(x)
		\label{eq:Hamiltonian}
	\end{eqnarray}
	with the single-particle Hamiltonian
	\begin{equation}\label{eq:single_particle_hamiltonian}
		\hat{h} = -\frac{\hbar^2}{2 m} \frac{d^2}{dx^2} + V(x)
	\end{equation}
	and $\hat{\Psi}_\sigma(x)$ -- the field operator annihilating a boson in state $\sigma$ at position $x$. The interaction coefficients $g_{\sigma\sigma'}=g^{3D}_{\sigma\sigma'}/(2 \pi L_\perp^2)$ are determined by the transverse confinement length $L_\perp$ and by the coupling constants in three dimensions $g^{3D}_{\sigma\sigma'}=4\pi \hbar^2 a_{\sigma\sigma'}/m$, where $a_{\sigma\sigma'}$ is the $s$-wave scattering length for one atom in the state $\sigma$ and the other one in $\sigma'$, $m$ is the atomic mass and $\hbar$ is the Planck constant. This study assumes repulsive interaction between atoms in the two states, with $A \leftrightarrow B$ symmetry leading to $g_{AA} = g_{BB}$, and an adjustable inter-species coupling $g_{AB}$, restricting to the phase-mixing regime: $g_{AB}<g_{AA}$ \cite{PhysRevLett.81.5718}.

	The motion of a particle in a periodic potential formed by an optical lattice can be conveniently described by the band theory, and the system Hamiltonian \eqref{eq:Hamiltonian} can be considered in the basis of Bloch functions $\psi_{l,q}(x, t)$~\cite{RevModPhys.80.885}. These are eigenfunctions of the single-particle Hamiltonian \eqref{eq:single_particle_hamiltonian} and posses required translational properties. The Bloch functions can be constructed numerically in the plane wave basis as described in \cite{paper}. As the next step, the Hamiltonian is rewritten in the basis of Wannier functions $w(x-x_i,t)$ localised around lattice sites, where $x_i$ denotes position of the $i\nth$ site, in the lowest energy band. The Wannier functions can be conveniently constructed from the Bloch states $\psi_{l=1,q}(x,t)$ in the following way
		\begin{equation}\label{eq:wannier_bloch}
		w(x-x_i,t) = \left( \frac{d}{2 \pi} \right) ^{1/2} \int\limits_{q \in 1BZ} {\rm d}q \, e^{-i q x_i } \psi_{l=1,q}(x,t).
		\end{equation}
		Summation in Equation~\eqref{eq:wannier_bloch} extends over wave vectors belonging to the 1\textsuperscript{st} \newterm[Brillouin zone] (1BZ), $-k_L< q \leq k_L$. The two-component Hamiltonian is then
		\begin{align}\label{eq:ham}
		\hat{\mathcal{H}} &= -\sum\limits_{i, j} J(i-j) \left(\hat{a}_{i}^{\dagger}\hat{a}_{j} + \hat{b}_{i}^{\dagger}\hat{b}_{j}\right) \nonumber\\
		+ & \frac{1}{2}\sum\limits_{i,j,k,l} U^{AA}_{i,j,k,l} \hat{a}^\dagger_i \hat{a}^\dagger_j \hat{a}_k \hat{a}_l 
		+ \frac{1}{2}\sum\limits_{i,j,k,l} U^{BB}_{i,j,k,l} \hat{b}^\dagger_i \hat{b}^\dagger_j \hat{b}_k \hat{b}_l \nonumber\\
		+ & \sum \limits_{i,j,k,l} U^{AB}_{i,j,k,l} \hat{a}^\dagger_i \hat{b}^\dagger_j \hat{a}_k \hat{b}_l ,
		\end{align}
		where $\hat{a}_i^{\dagger}$ (or $\hat{b}_i^{\dagger}$) creates a particle in the single-particle Wannier state $w(x-x_i,t)$ of the lowest energy band ($l=1$) localised on the $i\nth$ site, in the internal state $A$ \, (or $B$ respectively). The commutation relations between those operators are presented in the equations \eqref{eq:CommutationRelations}.
		
		The Bose-Hubbard model considers only states in the lowest energy band, which is justified as long as the excitation energies to the higher bands are much larger than the energies involved in the system dynamics. In general, if the lattice height $V_{0}(t)$ is being varied in time, the Wannier functions, and hence the hopping and interaction parameters depend on time
		\begin{subequations}\label{eq:hopping_inter}
			\begin{align}
			J(i-j) & = - \frac{d}{2\pi} \int\limits_{q \in 1BZ} {\rm d} q \, E_q e^{-i (i-j)d q}\, \label{eq:tunnelling}\\
			U^{\sigma\sigma'}_{i,j,k,l} & = g_{\sigma\sigma'} \int {\rm d}x\, w(i)w(j)w(k)w(l),
			\label{eq:UDefinition}
			\end{align}
		\end{subequations}
		where $w(i)=w(x-x_i)$. 
		Due to the orthonormality of the Wannier functions, the tunnelling terms $J(i-j)$ are independent on the dimension while the interaction parameters are not, as the Wannier function factorises in the respective spatial dimensions.
	
	In the \newterm[tight-binding] limit when the lattice height is larger than the recoil energy $E_R=\hbar^2k_L^2/(2m)$ and the Wannier functions are well-localised around each lattice site, the tunnelling and interaction terms fall off rapidly with the distance $|x_i - x_j|$. By keeping  only the leading terms in \eqref{eq:hopping_inter}, one obtains the Bose-Hubbard model
	\begin{gather} \label{Hamiltonian}
	\hat{\mathcal{H}}_{BH} = \hat{\mathcal{H}}_{hopp} + \hat{\mathcal{H}}_{int},
	\end{gather}
	where
	\begin{subequations}
		\begin{align}
		\label{hopping_term}
		\hat{\mathcal{H}}_{hopp} &= - J\sum_{\left< i,j\right> =1}^{M} \hat{a}_i^{\dagger} \hat{a}_{j}  - J\sum_{\left< i,j\right> =1}^{M} \hat{b}_i^{\dagger} \hat{b}_{j}, \\
		\label{interaction_term}
		\hat{\mathcal{H}}_{int} &= \frac{U_{AA}}{2} \sum_{i=1}^{M} \hat{n}^A_i (\hat{n}^A_i - 1) + \frac{U_{BB}}{2} \sum_{i=1}^{M} \hat{n}^B_i (\hat{n}^B_i - 1) + U_{AB} \sum_{i=1}^{M} \hat{n}^A_i \hat{n}^B_i.
		\end{align}
	\end{subequations}
	{In the above expressions $J:=J(1)$, $U_{\sigma\sigma'}:=U^{\sigma\sigma'}_{0,0,0,0}$ and $\hat{n}^A_i:=\hat{a}_{i}^{\dagger}\hat{a}_{i}$, $\hat{n}^B_i:=\hat{b}_{i}^{\dagger}\hat{b}_{i}$. In the shallow lattice, however, for $V_0 \ll 1$, all the terms in \eqref{eq:ham} have to be taken into account. A schematic illustration of the Bose-Hubbard model is depicted in Figure~\ref{BH_figure}.}

	\begin{figure}[H]
		\centering
		\includegraphics[width=0.8\textwidth]{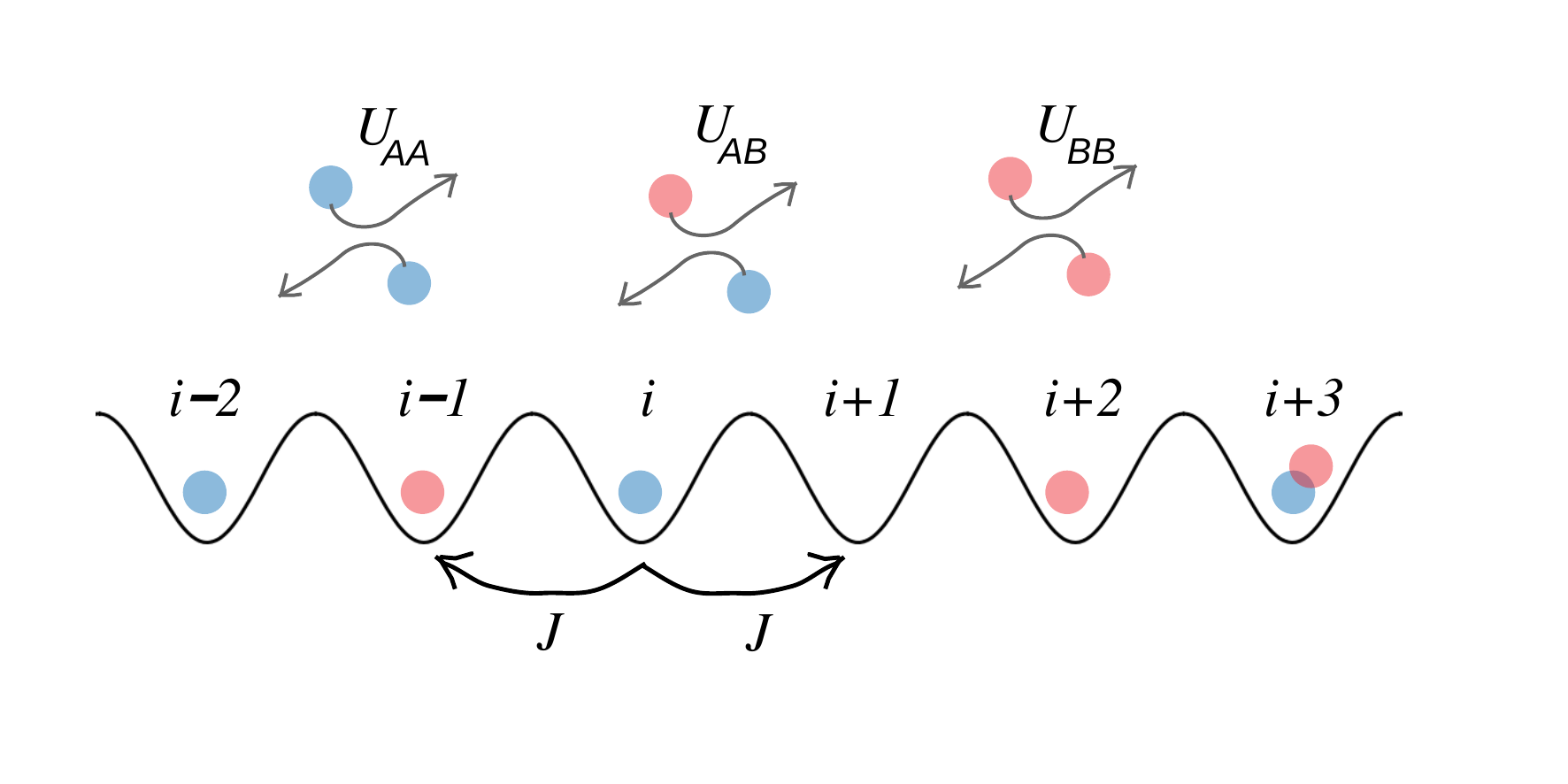}
		\caption{
		{A scheme of the Bose-Hubbard model for $N$ bosonic atoms in two internal states $A$ and $B$ (blue and red) in the ground band of the optical lattice (solid black line). Atoms can tunnel to neighbour sites at the rate $J$ and interact with the strength $U_{AA}, U_{AB}, U_{BB}$. Adapted from \cite{paper}.}}
		\label{BH_figure}
	\end{figure}

	The first and second term in \eqref{Hamiltonian} is called \newterm[hopping term] \eqref{hopping_term} and \newterm[interaction term] \eqref{interaction_term} respectively. The sums in the hopping term go over neighbouring indices ($|i-j|\,=1$). The $\hat{a}_i^{\dagger}$ and $\hat{b}_i^{\dagger}$ operators are \newterm[creation operators] of an atom in states $A$ and $B$ on $i\nth$ of $M$ lattice's sites. They obey commutation relations
	\begin{subequations}
	    \label{eq:CommutationRelations}
		\begin{align}
		[\hat{\alpha}_i,\hat{\beta}^\dagger_j] = \delta_{ij} \delta_{\alpha\beta}, \\
		[\hat{\alpha}^\dagger_i,\hat{\beta}^\dagger_j] = [\hat{\alpha}_i,\hat{\beta}_j] = 0,
		\end{align}
	\end{subequations}
	for $\alpha = a, b$ and $\beta = a, b$. The \newterm[atoms number operators] $\hat{n}_i^\alpha$ are defined as:
	\begin{equation}
	\hat{n}^A_i := \hat{a}_i^{\dagger} \hat{a}_i, ~~~\hat{n}^B_i := \hat{b}_i^{\dagger} \hat{b}_i,
	\end{equation}
	and \newterm[total atoms number operator] $\hat{N}$ is defined as:
	\begin{equation}
	\hat{N} := \hat{N}^{A} + \hat{N}^{B} := \sum_{i=1}^{M} \hat{n}_i^A + \sum_{i=1}^{M} \hat{n}_i^B,
	\end{equation}
	where $\hat{N}^{\alpha}:=\sum_{i=1}^{M} \hat{n}_i^\alpha$ are total atoms number operators of type $\alpha$.
	
	The $U_{AA}, U_{BB}, U_{AB}$ and $J$ parameters depend on $V_0$ -- depth of the lattice potential. A more detailed discussion of the dependence of values $J, U_{AA}, U_{BB}, U_{AB}$ on $V_0$ is presented in Section~\refsec[sec:deepLattice]. In the proposed protocol, parameters of the Bose-Hubbard model are time-dependent, as a change of the lattice height $V_0$ is induced. Changing $V_0$ is necessary to cross the system from the superfluid to the Mott insulator phase. 

	In the simulations, the symmetry between type $A$ and $B$ particles is assumed and:
	\begin{equation}
	U := U_{AA}(V_0) = U_{BB}(V_0) = 0.95^{-1}\, U_{AB}(V_0) \geq 0.
	\end{equation}

	For $V_0$ a linear dependence on time is considered:
	\begin{equation}
	V_0(t) = V_i + (V_f - V_i)\frac{t}{\tau}, 
	\end{equation}
	where the lattice depth $V_0$ changes from $V_i$ to $V_f$ in time $\tau$, which is referred to as \newterm[ramp time].

	\section{The notation}
	
    \subsection{Fock basis}
    \label{sec:FockBasis}

	This work uses the Fock basis $\mathcal{F}$ to describe the system. Any element of the Fock basis of an optical lattice with $M$ lattice sites and $N$ particles of two types ($A$ and $B$) in the considered model can be written in the form
	\begin{equation} \label{eq:FockState}
	\ket{n^A_{1}; n^B_{1}} \otimes \ket{n^A_{2}; n^B_{2}} \otimes \cdots \otimes \ket{n^A_{M}; n^B_{M}} \equiv \ket{n^A_{1}, ..., n^A_{M}; n^B_{1}, ..., n^B_{M}},
	\end{equation}
	where $n^\alpha_i$ is the number of $\alpha$ type particles on $i\nth$ lattice site, so
	\begin{equation}
	\hat{n}^\alpha_i \ket{n^A_{1}, ..., n^A_{M}; n^B_{1}, ..., n^B_{M}} = n^\alpha_i \ket{n^A_{1}, ..., n^A_{M}; n^B_{1}, ..., n^B_{M}},
	\end{equation}
	\begin{equation}
	\sum_{\alpha=A,B} \sum_{i=1}^M n^\alpha_i = \sum_{\alpha=A,B} \! N^\alpha = N
	\end{equation}
	and creation/annihilation operators for type $A$ work so that
	\begin{subequations}
		\begin{align}
		\hat{a}_i \ket{n^A_{1}, ..., n^A_{i}, ..., n^A_{M}; \{ n_i^B \}} = \sqrt{n^A_{i}-1} \ket{n^A_{1}, ..., n^A_{i}-1, ..., n^A_{M}; \{ n_i^B \}},\\
		\hat{a}^\dagger_i \ket{n^A_{1}, ..., n^A_{i}, ..., n^A_{M}; \{ n_i^B \}} = \sqrt{n^A_{i}} \ket{n^A_{1}, ..., n^A_{i}+1, ..., n^A_{M}; \{ n_i^B \}}
		\end{align}
	\end{subequations}
	(where $\{ n_i^B \} = n^B_{1}, ..., n^B_{M}$), analogically for type $B$.
	
	Any state $\ket{\Psi(t)}$ can be then written in the form:
	\begin{equation}
	\ket{\Psi(t)} \equiv \sum_{\psi\in\mathcal{F}} \ket{\psi}\bra{\psi} \ket{\Psi(t)} =: \sum_{\psi\in\mathcal{F}} c_\psi (t) \ket{\psi},
	\end{equation}
	where the complex numbers $c_\psi (t):=\braket{\psi|\Psi(t)}$ are called decomposition coefficients of the state $\ket{\Psi(t)}$ in the Fock basis $\mathcal{F}$.

	Similarly, any operator acting on the whole lattice is also considered in the Fock basis:
	\begin{equation}
	\hat{\mathcal{O}} = \sum_{\psi,\psi'\in\mathcal{F}} \ket{\psi}\bra{\psi} \hat{\mathcal{O}} \ket{\psi'}\bra{\psi'} =: \sum_{\psi,\psi'\in\mathcal{F}} \mathcal{O}_{\psi,\psi'} \ket{\psi}\bra{\psi'}
	\end{equation}
	and the complex numbers $\mathcal{O}_{\psi,\psi'}$ are the elements of the matrix representing the operator $\hat{\mathcal{O}}$ in the Fock basis.

	It is the vector/matrix representation, where states are represented by vectors while operators by matrices. In this way, any operation (like the calculation of an average) can be translated to a numerical multiplication or addition of matrices and vectors.

	\subsection{Single-site operators}
	\label{sec:Notation}

	In this study, quantum states describing multiple lattice sites are examined. Sometimes, the operators used in the work act on the whole state, however, other times they act only on a single lattice. It is thus convenient to introduce the notation below. An operator acting on a single site state is marked by a superscript with parentheses. For instance, $\hat{\mathcal{O}}^{(j)}$ operator acts on a state describing $j\nth$ lattice site only. For any operator $\hat{\mathcal{O}}^{(j)}$ acting on a state describing a single lattice site, one can define an operator $\hat{\mathcal{O}}_j$ acting on a state describing the whole lattice such that
	\begin{equation}
        \hat{\mathcal{O}}_j := \identity^{(1)} \otimes \cdots \otimes \identity^{(j-1)} \otimes \hat{\mathcal{O}}^{(j)} \otimes \identity^{(j+1)} \otimes \cdots \otimes \identity^{(M)}
	\end{equation}
	and $\identity^{(k)}$ are identity operators acting on state describing $k$\nth~ lattice site.
	One can note here that the mathematical form of $\hat{\mathcal{O}}^{(j)}$ operators with different values of $j$ is the same. However, by changing the value of $j$ in the superscript of the operator $\hat{\mathcal{O}}^{(j)}$ one changes the physical interpretation of the operator, because it acts on another lattice site. On the other hand, changing the lower index of $\hat{\mathcal{O}}_j$, changes also the mathematical form of this operator.

	This thesis uses also operators of form:
	\begin{equation}
	    \label{eq:OperatorSum}
        \hat{\mathcal{O}} = \sum_{j=1}^M \hat{\mathcal{O}}_j,
	\end{equation}
	like the spin operators defined in Subsection~\refsec[sec:SpinOperators] or:
	\begin{equation}
	    \label{eq:OperatorProduct}
        \hat{\mathcal{O}} = \bigotimes_{j=1}^M \hat{\mathcal{O}}^{(j)},
	\end{equation}
	like the $\hat{S}_{rot}$ operator defined in Section~\refsec[sec:CorrelatorChoice]. Operators of the structure \eqref{eq:OperatorSum} and \eqref{eq:OperatorProduct} act on the state describing the whole lattice applying the same operation to each of the lattice sites.

	\subsection{Spin operators}
	\label{sec:SpinOperators}
	
	This work utilises the below defined collective spin operators $\hat{S}_x$, $\hat{S}_y$, $\hat{S}_z$:
	\begin{subequations}
	\label{eq:SpinOperators}
	    \begin{align}
	    \hat{S}_x := \sum_{j=1}^{M} \frac{1}{2} (\hat{a}^{\dagger}_j\hat{b}_j + \hat{b}^{\dagger}_j\hat{a}_j)\\
    	\hat{S}_y := \sum_{j=1}^{M} \frac{1}{2i} (\hat{a}^{\dagger}_j\hat{b}_j - \hat{b}^{\dagger}_j\hat{a}_j) \\
    	\hat{S}_z := \sum_{j=1}^{M} \frac{1}{2} (\hat{a}^{\dagger}_j\hat{a}_j - \hat{b}^{\dagger}_j\hat{b}_j).
    	\end{align}
	\end{subequations}
	The operators obey cyclic commutation relations:
	\begin{equation}
	    \left[ \hat{S}_j, \hat{S}_k \right] = \epsilon_{jkl} ~i \hat{S}_l,
	\end{equation}
	with $\epsilon_{jkl}$ being the Levi-Civita symbol.
	
	Note that the average values of $\hat{S}_x$, $\hat{S}_y$, $\hat{S}_z$ operators span a generalised Bloch sphere of radius $\sfrac{N}{2}$. It means that the length of the \newterm[collective spin] $\braket{\overrightarrow{\hat{S}}}$ is
	\begin{equation}
	   \left| \braket{\overrightarrow{\hat{S}}} \right| := \sqrt{\braket{\hat{S}_x}^2+\braket{\hat{S}_y}^2+\braket{\hat{S}_z}^2} = \frac{N}{2}.
	\end{equation}
	The Hamiltonian \eqref{Hamiltonian} conserves the number of particles; thus, the collective spin is also conserved during the evolution.

	\section{Very deep lattices -- harmonic approximation} \label{sec:deepLattice}

	{In the limit of $V_0/E_R \gg 1$, the Wannier functions can be approximated by Gaussian functions whose width is set by the frequency associated with each lattice site minimum~\cite{RevModPhys.80.885}. Therefore, $w(x)\approx \left( \frac{k_L^2}{\pi} \right)^{1/4} \left( \frac{V_{0}}{E_R} \right)^{1/8} e^{-\sqrt{V_{0}} k_L^2 x^2/2}$. This is a sufficient approximation to obtain the interaction coefficient
		\begin{equation}
		\frac{U^{\rm Gauss}_{\sigma\sigma'} (t) }{E_R} \approx \sqrt{\frac{32}{\pi}} \frac{a_{\sigma\sigma'} d}{L_x L_y} \left(\frac{V_{0}}{E_R}\right)^{1/4} . \label{eq:GaussU}
		\end{equation}
		Figure~\ref{UsandJs_plots}~(a) presents the interaction terms calculated exactly from \eqref{eq:UDefinition}, compared with the approximated formula \eqref{eq:GaussU}. It can be easily observed that for $V_0/E_R>1$, the interaction terms beyond the terms involving nearest neighbours can be neglected as compared with the latter ones as expected~\cite{RevModPhys.80.885}.}

\begin{figure}[H]
	\centering
	\includegraphics[width=0.49\textwidth]{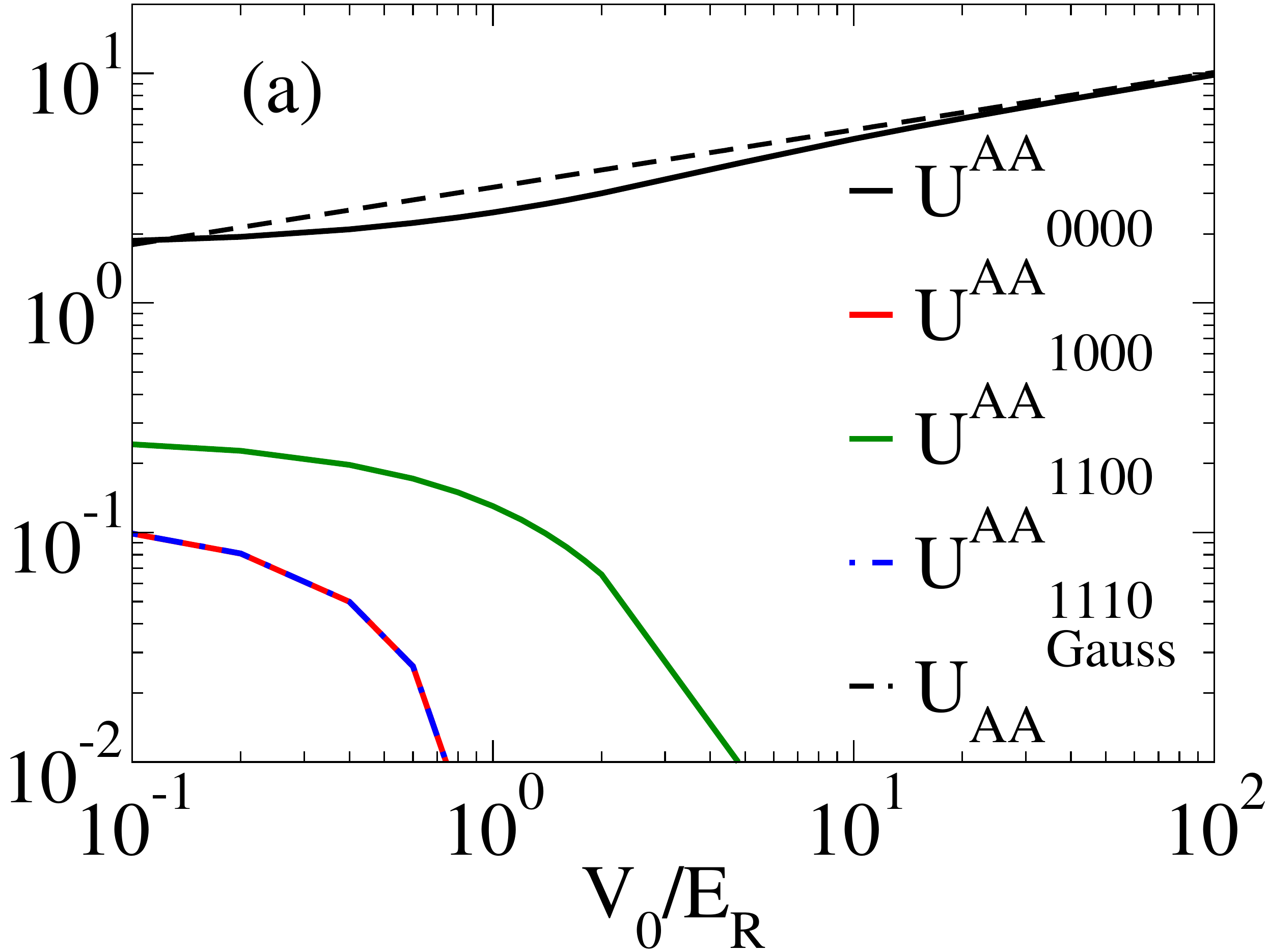}
	\includegraphics[width=0.49\textwidth]{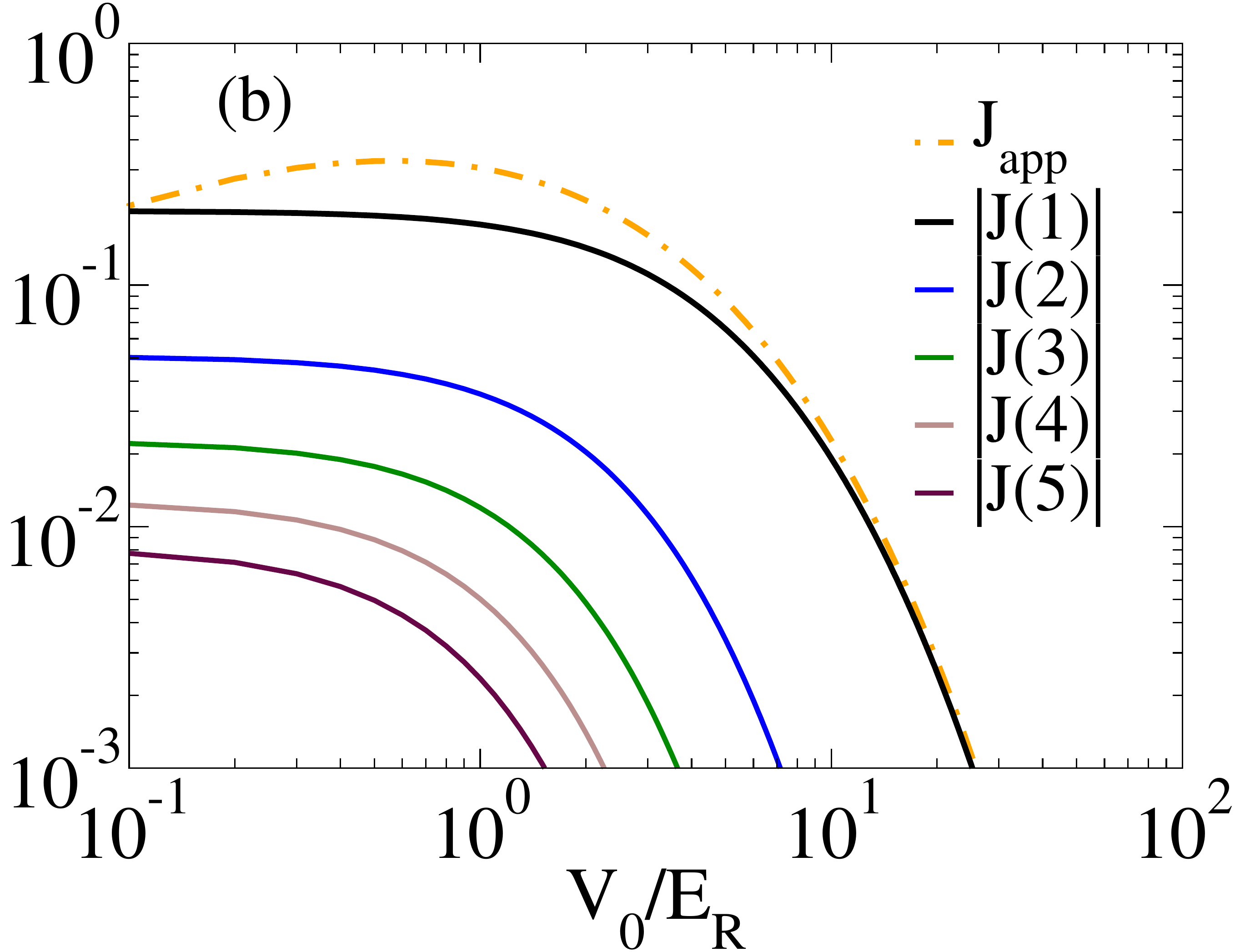}
	\caption{(a) The interaction terms $U^{AA}_{i,j,k,l}/\left[ a_{\sigma\sigma'}d/(L_x L_y) \right]$ \eqref{eq:UDefinition} calculated numerically using the Wannier function for the values of $i,j,k,l$ as indicated in the legend (solid lines). The black dashed line shows the result of the Gaussian approximation for $U^{AA}_{0000}/\left[ a_{\sigma\sigma'}d/(L_x L_y) \right]$~(\ref{eq:GaussU}).
	(b) The tunnelling terms in units of the recoil energy $E_R$ calculated exactly from the energy spectrum of the single-particle Hamiltonian in the case of different distances $i-j$ between sites as indicated in the legend. The approximated result for the nearest neighbour~(\ref{eq:GaussJ}) is shown by the orange dot-dashed line. Reprinted from \cite{paper}.}
	\label{UsandJs_plots}
\end{figure}

	{In the case of the hopping parameters $J(i-j)$, the Gaussian approximation gives a relative error growing with the lattice height~\cite{PhysRevA.79.053623}. The hopping matrix element can be conveniently approximated by the width of the lowest band in the one-dimensional Mathieu equation~\cite{BH_model}, yielding to
	\begin{equation}
		\frac{J_{\rm app}}{E_R} \approx \frac{4}{\sqrt{\pi}} \left(\frac{V_{0,x}}{E_R}\right)^{3/4} e^{-2\sqrt{\frac{V_{0,x}}{E_R}}}, \label{eq:GaussJ}
	\end{equation}
	for the nearest neighbours, i.e. $\left|i-j\right|=1$. 
	Figure~\ref{UsandJs_plots}~(b) shows the tunnelling terms calculated exactly and using the above formula. Indeed, as long as $V_0/E_R>1$, the leading nearest neighbour term dominates.}

	This work operates in the regime $3 < V_0/E_R < 40$; thus, the terms other than $U^{\sigma\sigma'}_{0000} =: U_{\sigma\sigma'}$ and $J(1) =: J$ are omitted. The values used in the numerical calculations are computed from \eqref{eq:UDefinition} and \eqref{eq:tunnelling}. They are presented in Figure~\ref{UandJ_plot}.

\begin{figure}[H]
	\centering
	\includegraphics[width=0.8\textwidth]{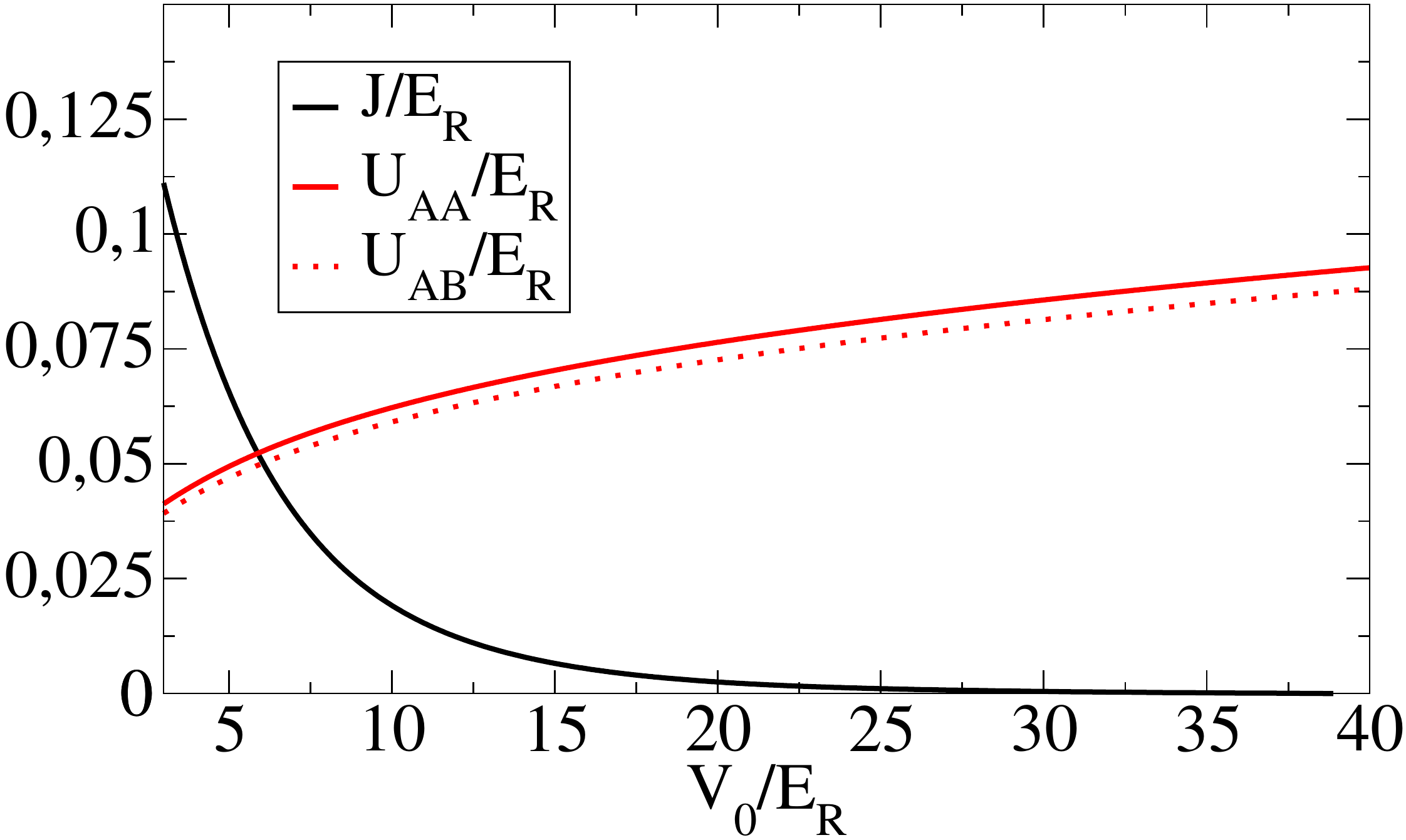}
	\caption{Values of $U_{AA}=U_{BB}$, $U_{AB}$ and $J=J(1)$ in dependence of $V_0$, calculated from \eqref{eq:UDefinition} and \eqref{eq:tunnelling}, used in the simulations.}
	\label{UandJ_plot}
\end{figure}

	\section{Phase diagram of ground states} \label{sec:phaseDiagram}

	The Bose-Hubbard Hamiltonian describes many-body system. One can then investigate its thermodynamic properties. As this thesis deals with systems with finite number of particles, the thermodynamic limit cannot be achieved. However, for simplicity, the terminology from statistical physics is used, because in finite systems there appear analogous phenomena as in thermodynamics of infinite systems.

    The following subsections describe two interesting regimes of the phase diagram of the ground states of the Hamiltonian. Those are Mott insulator regime and superfluid regime.

	In order to calculate the phase diagram at zero temperature, a numerical program to compute the ground states of the Bose-Hubbard Hamiltonian for different numbers of particles per site $\nu:=N/M$ and $U/J$ ratios was prepared. For $M = 8$ sites in the lattice and $U_{AA} = U_{BB}$, $U_{AB} = 0.95 \, U_{AA}$, the variance of particle numbers and condensate fraction were calculated for different values of $J$ and different total numbers of atoms $N$. The results are presented in subsections \refsec[subs:MI] and \refsec[subs:SF].

	\begin{quote}
		\subsection{Mott insulator phase}
		\label{subs:MI}
		
		In the regime, where $J \ll U_{\sigma\sigma'}$, the hopping term becomes negligible. The atoms become localised as probability of hopping between different lattice sites becomes very small.

		The interaction between atoms is repulsive ($U_{\sigma\sigma'} \geq 0$) and the system is translationally invariant; thus, the system tends to occupy the states with homogeneous distribution of atoms over the lattice. The ground states of the system are realised with exactly one atom per site for $N=M$.

		This regime is called the \newterm[Mott insulator phase] \cite{BH_phases}. It is interesting as in this regime, the interaction between atoms at neighbouring sites is suppressed and many-body properties such as entanglement can be ``frozen''.

		To verify if the system is in the Mott insulator phase, the variance of the population of particles in each site is used:
		\begin{equation} \label{var_i}
		\var_i := \Delta^2 (\hat{n}^A_i + \hat{n}^B_i) \equiv \braket{(\hat{n}^A_i + \hat{n}^B_i)^2}-\braket{\hat{n}^A_i + \hat{n}^B_i}^2.
		\end{equation}
		It tends to zero for each site $i$ in the Mott regime. The diagram of $\delta_1$ is plotted in Figure~\ref{phaseDiag-var}.
		
		\begin{figure}[h!]
			\centering
			\includegraphics[width=0.8\textwidth]{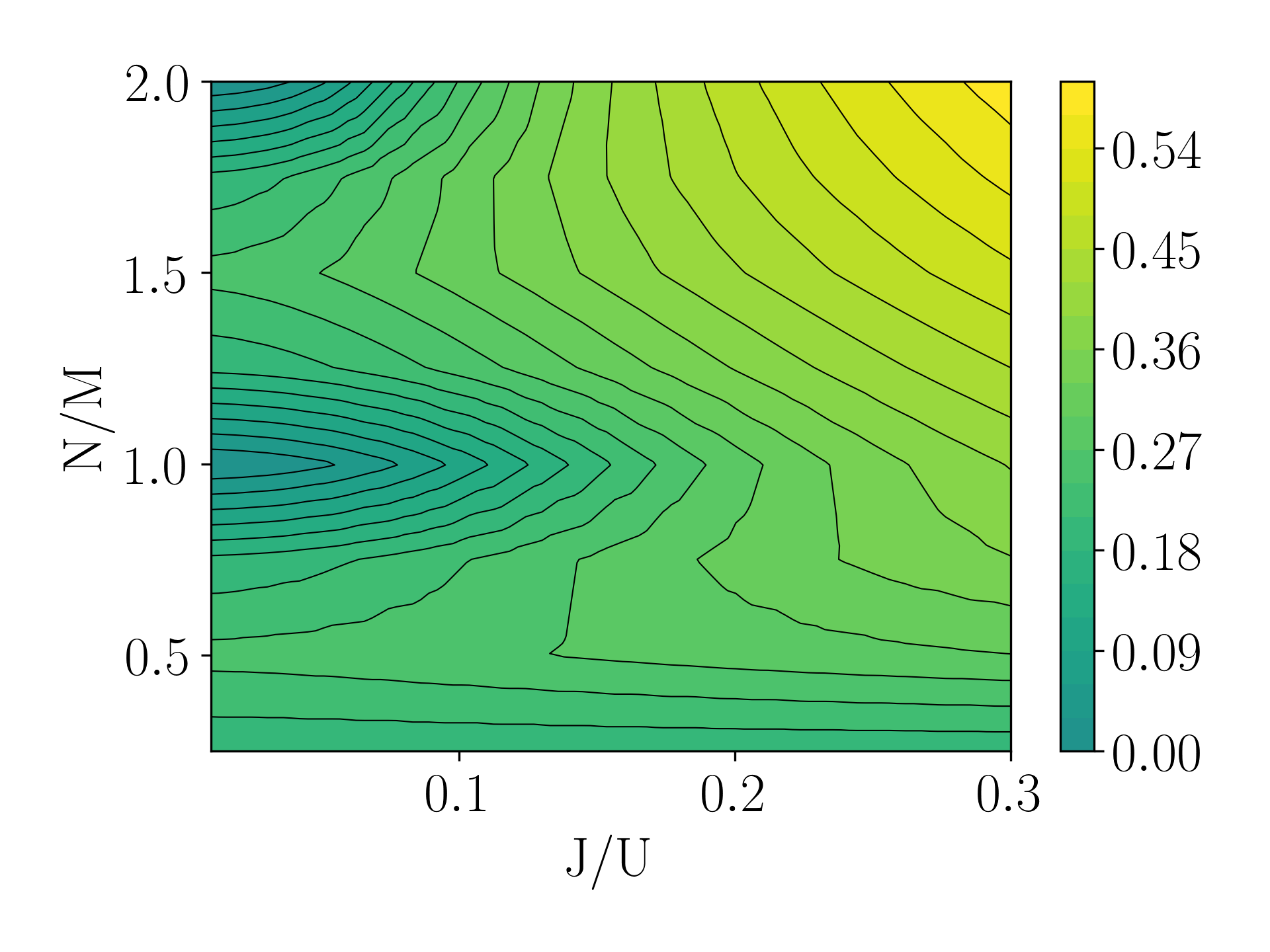}
			\caption{A plot showing interpolated values of $\var_1$ (\ref{var_i}) versus $N/M$ and different values of $J/U_{\sigma\sigma}$. One can observe two spots of Mott insulator regime (characterised by $\var_1\approx0$) for small values of $J/U_{\sigma\sigma}$ and integer number of particles per site $\nu=N/M$. The values of $\var_i$ for different sites $i$ are equal due to the symmetry of the problem.}
			\label{phaseDiag-var}
		\end{figure}

		\subsection{Superfluid phase}
		\label{subs:SF}
		
		The hopping term becomes the most important in the Hamiltonian when $J \gg U_{\sigma\sigma'}$. In this regime the system is \newterm[superfluid] -- the wave functions of atoms are spatially delocalised and atoms can easily interact with each other \cite{BH_phases}.

		This study employs the notion of \newterm[superfluid fraction] to indicate the existence of the superfluid phase. In the case of this work, the superfluid fraction is equal to \newterm[condensate fraction] \cite{SuperfluidFraction}
		\begin{equation} \label{eq:Fraction}
		f_c := \frac{1}{N M} \sum_{i=1}^M \sum_{j=1}^M \sum_{\alpha=a,b} \braket{\hat{\alpha}_i^{\dagger} \hat{\alpha}_j},
		\end{equation}
		which measures the amount of atoms in the zero quasi-momentum mode.
		
		Figure \ref{phaseDiag-frac} shows interpolated values of $f_c$ -- a measure of condensate fraction and superfluid fraction as well -- depending on $J/U_{\sigma\sigma}$ and $N/M$.
		
		\begin{figure}[H]
			\centering
			\includegraphics[width=0.8\textwidth]{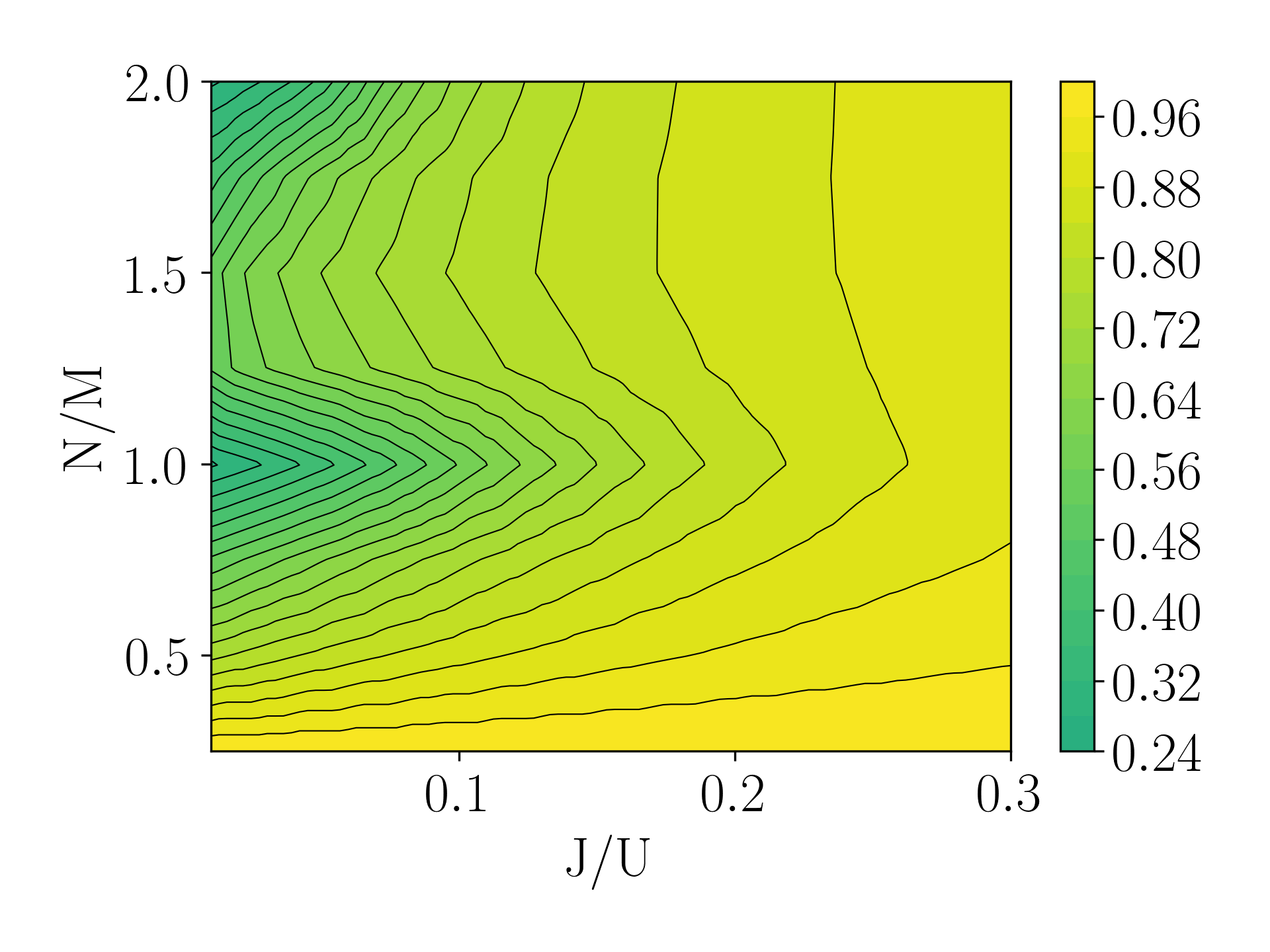}
			\caption{A graph of interpolated values of $f_c$ as a function of the average number of particles per site $N/M$ and the value of $J/U_{\sigma\sigma}$. It is observed that for larger values of $J$, in the region marked by yellow colour, the system is in the superfluid regime as the condensate fraction is significant.}
			\label{phaseDiag-frac}
		\end{figure}

	\end{quote}

	\section{Entanglement} 
    \label{sec:QuantumEntanglement}

	Entanglement is one of the key features of quantum systems. As being unique for quantum mechanics, it witnesses non-classicality of physical systems \cite{Bell}. It can also be used as a resource in quantum information protocols such as dense coding or quantum teleportation \cite{photonic_QI, entanglement_review}. 

	In the linear-algebraic terms of quantum theory, entanglement is a property of a multipartite system, which is described by non-separable density operator. To explain what is non-separability, first definition of \newterm[separability] has to be introduced. A quantum state described by the density operator $\hat{\rho}_{sep}$ is \newterm[separable] if the operator can be written in the form
	\begin{equation}
	\label{separability}
	\hat{\rho}_{sep} = \int d\lambda~p(\lambda) \bigotimes_{j=1}^M \hat{\rho}^{(j)}(\lambda),
	\end{equation}
	where $\hat{\rho}^{(j)}(\lambda)$ is a density operator describing state of $j$\nth ~of $M$ parts of the system and $p(\lambda)$ is probability distribution over some random variable $\lambda$.	If a state cannot be described in a form \eqref{separability}, it is called \newterm[non-separable] or \newterm[entangled].
	
	This work considers entanglement between modes as the parts of the system are lattice sites. If the examined parts of the system were particles, then one would deal with entanglement between particles.

	\subsection{Quantum correlator \Cesq}
	\label{subs:Correlator}
		
	The goal of this thesis is to produce an entangled state; thus, it has to be justified whether created states are entangled. Following the discussion from \cite{Chwe20}, a similar measure of entanglement is used in this work -- the quantum correlator \Cesq defined below:
	\begin{equation} \label{eq:GeneralCorrelator}
	\left| \mathcal{C}_{\vec{e}} \right|^2 := \left| \braket{\bigotimes_{j=1}^M \hat{S}_{\vec{e}}^{(j)}} \right|^2.
	\end{equation}
	The $\hat{S}_{\vec{e}}^{(j)}$ operators are spin operators in $\vec{e}$ direction acting on $j\nth$ lattice site, defined as
	
	\begin{equation}
	\label{eq:SpinEOperator}
	\hat{S}_{\vec{e}}^{(j)} := \vec{e} \cdot \vec{\hat{S}}^{(j)} \equiv \vec{e} \cdot (\hat{S}_x^{(j)}, \hat{S}_y^{(j)}, \hat{S}_z^{(j)})^T,
	\end{equation}
	where $\vec{e}$ is a three-dimensional vector of complex numbers $e_x, e_y, e_z$ (independent of $j$) such that its scalar product with any unit vector of real numbers $\vec{n}$ is no greater than one:
	\begin{equation}
	        \label{eq:VecEConditions}
	        \forall_{ \substack{ \vec{n}:\\ |\vec{n}|=1 }} \left| \vec{e}\cdot \vec{n} \right| \leq 1
	\end{equation}
	(the choice of the above condition is explained in Appendix~\refsec[app:CumbersomeInequality]) and
	\begin{subequations}
	\label{eq:Pauli}
		\begin{align}
		\hat{S}_x^{(j)} := \frac{1}{2} \left( \hat{a}^{{(j)}^\dagger} \hat{b}^{(j)} + \hat{b}^{{(j)}^\dagger} \hat{a}^{(j)} \right), \\
		\hat{S}_y^{(j)} := \frac{1}{2\I} \left( \hat{a}^{{(j)}^\dagger} \hat{b}^{(j)} - \hat{b}^{{(j)}^\dagger} \hat{a}^{(j)} \right), \\
		\hat{S}_z^{(j)} := \frac{1}{2} \left( \hat{a}^{{(j)}^\dagger} \hat{a}^{(j)} - \hat{b}^{{(j)}^\dagger} \hat{b}^{(j)} \right).
		\end{align}
	\end{subequations}

	What is worth mentioning is that the value of \Cesq calculated on a mixed state with $N=M$ particles in $M$ lattice sites is limited by the following bound:
	\begin{equation}
	\left| \mathcal{C}_{\vec{e}} \right|^2 \leq \frac{1}{4},
	\end{equation}
	which is proved in Appendix~\refsec[app:MaxCorrelator].
    
    It has to be noted that the definition \eqref{eq:GeneralCorrelator} of the correlator, represents a whole class of possible correlators, which differ by the direction of $\vec{e}$. Application of the correlator as a witness of entanglement in a given physical problem requires a careful choice of the direction $\vec{e}$. In the context of the Bose-Hubbard model it is shown in Section~\refsec[sec:CorrelatorChoice]. A special case is, e.g. $\left| \mathcal{C}_{+} \right|^2~$ correlator with $\hat{S}_{\vec{e}}^{(j)} = \hat{S}_+^{(j)} := \hat{S}_x^{(j)} + \I \hat{S}_y^{(j)}$, which was used in \cite{Chwe20} and in this work in Section~\refsec[sec:CorrelatorChoice].

	\subsection{\Cesq as a witness of entanglement} \label{subs:Witness}

	In the thesis, states with equal number of particles and lattice sites ($N=M$) are investigated. For these states, Condition \eqref{eq:Bound} necessary for separability of modes is presented and proven as below.
	
	Assuming that a density matrix $\hat{\rho}_{sep}$ describing some system is separable, one can write as follows:
	\begin{dmath}
		\label{eq:BoundInequalities}
		\left| \mathcal{C}_{\vec{e}} \right|^2 \equiv
		\left| \Tr \left[ \hat{\rho}_{sep} \bigotimes_{j=1}^M \hat{S}_{\vec{e}}^{(j)} \right] \right|^2\,\text{=} \overset{\dgr[1]}{=}
		\left| \int d\lambda~ p(\lambda) \prod_{j=1}^M \Tr \left[ \hat{\rho}^{(j)}(\lambda) \hat{S}_{\vec{e}}^{(j)} \right] \right|^2\,\text{$\leq$}  \overset{\dgr[2]}{\leq}
		\int d\lambda~ p(\lambda) \prod_{j=1}^M \left| \Tr \left[ \hat{\rho}^{(j)}(\lambda) \hat{S}_{\vec{e}}^{(j)} \right] \right|^2\,\text{$\leq$} \overset{\dgr[3]}{\leq} 2^{-2M}.
	\end{dmath}
	The above holds since: $\dgr[1]$ the state is separable (\eqref{separability}), $\dgr[2]$ a square of a modulus of mean value of a random variable is no greater than a mean value of square of modulus of the random variable\footnote{For any complex random variable $X$ the following holds: $\braket{\left(|X|-|\braket{X}|\right)^2}\geq0$. As a consequence $|\braket{X}|^2\leq\braket{|X|^2}$ because $\braket{|X|}=|\braket{X}|$.} and $\dgr[3]$ is just a special case of the equation:
	\begin{equation}
	    \int d\lambda~ p(\lambda) \prod_{j=1}^M \left| \Tr\left[ \hat{\rho}^{(j)}(\lambda) \hat{S}_{\vec{e}}^{(j)} \right] \right|^2 \leq \left(\frac{N}{2 M}\right)^{2M}
	\end{equation}
	with $N=M$, which is proved as Theorem~\refsec[the:Theorem1] in Appendix~\refsec[app:CumbersomeInequality].
	
	From the above it is apparent that for any separable state with $N = M$, the value of $\left| \mathcal{C}_{\vec{e}} \right|^2$ is less than $2^{-2M}$ for all directions $\vec{e}$:
	\begin{equation}
	    \label{eq:Bound}
	    ~~~~~~~~~~~~~~~~~~~~~~~~~~~~~~~~~~~~~~~~~~~ \left| \mathcal{C}_{\vec{e}} \right|^2 \leq 2^{-2M} \qquad\text{(for separable states)}.
	\end{equation}
	The value $2^{-2M}$ in Inequality~\eqref{eq:Bound} is called \newterm[separability bound] in this thesis, because breaking it would prove that the state cannot be considered a separable state. It would indicate that there is entanglement between the modes in the system. This work uses $\left| \mathcal{C}_{\vec{e}} \right|^2$ as a measure of entanglement in a multi-mode system. All the states with \Cesq$>2^{-2M}$ are \newterm[entangled] and the states with \Cesq$=\sfrac{1}{4}$ are called \newterm[fully entangled], because $\sfrac{1}{4}$ is the maximum possible value of \Cesq as proved in \cite{Chwe20} for $|\mathcal{C_+}|^2$ and in Appendix~\refsec[app:MaxCorrelator] for any $\vec{e}$. The value 
	
	It is worth noting that the state is entangled if for any $\vec{e}$ occurs \Cesq$>2^{-2M}$. On the other hand, for a given entangled state there may be such vector $\vec{e}$ for which the separability bound \eqref{eq:Bound} is not broken. To prove that the state is entangled one needs to find a proper $\vec{e}$ as done in Section~\refsec[sec:CorrelatorChoice]. Moreover, if one finds that the value of \Cesq is lower than $2^{-2M}$, for all $\vec{e}$ allowed by Condition \eqref{eq:VecEConditions}, it does not have to mean that the state is separable.

	\section{Generating entanglement in optical lattice} \label{sec:HamiltonianIn0q}
    
    This subsection briefly argues why entanglement can be generated in the system studied in this work. To give a better understanding, first, the Hamiltonian of the model is presented in the Fourier space. Then it is shown that the model takes the form of the one-axis twisting model and the consequences of this fact are discussed.
    
    The Bose-Hubbard Hamiltonian \eqref{Hamiltonian} can be written in the momentum representation as follows (see, e.g. \cite{kajtoch2018adiabaticity}):
    \begin{equation*}
        \hat{\mathcal{H}}_{BH} = \sum_q \epsilon_q \hat{\tilde{a}}_q^{\dagger} \hat{\tilde{a}}_q + \sum_q \epsilon_q \hat{\tilde{b}}_q^{\dagger} \hat{\tilde{b}}_q +
    \end{equation*}
    \begin{equation} \label{eq:MomentumHamiltonian}
        + \sum_{q_1, q_2, k} \left[ \frac{U_{AA}}{2M} \hat{\tilde{a}}_{q_1-k}^{\dagger} \hat{\tilde{a}}_{q_2+k}^{\dagger} \hat{\tilde{a}}_{q_1} \hat{\tilde{a}}_{q_2} + \frac{U_{BB}}{2M} \hat{\tilde{b}}_{q_1-k}^{\dagger} \hat{\tilde{b}}_{q_2+k}^{\dagger} \hat{\tilde{b}}_{q_1} \hat{\tilde{b}}_{q_2}  \right] +
    \end{equation}
    \begin{equation*}
        + \sum_{q_1, q_2, k} \frac{U_{AB}}{M} \hat{\tilde{a}}_{q_1-k}^{\dagger} \hat{\tilde{b}}_{q_2+k}^{\dagger} \hat{\tilde{a}}_{q_1} \hat{\tilde{b}}_{q_2},
    \end{equation*}
    where $\hat{\tilde{a}}_q^{\dagger}$ and $\hat{\tilde{b}}_q^{\dagger}$ are creation operators of a particle with quasi-momentum $q$ in state $A$ and $B$ respectively:
    \begin{subequations}
        \begin{align}
            \hat{a}_{i}^{\dagger} &= \frac1{\sqrt{M}} \sum_q e^{- \I q x_i} \hat{\tilde{a}}_q^{\dagger}, \\
            \hat{b}_{i}^{\dagger} &= \frac1{\sqrt{M}} \sum_q e^{- \I q x_i} \hat{\tilde{b}}_q^{\dagger}
        \end{align}
    \end{subequations}
    with $x_i$ being the position of $i\nth$ lattice site and
    \begin{equation}
        \epsilon_q = - 2 J \cos{q d},
    \end{equation}
    where $d$ is the spatial distance separating lattice sites.
		
	In the superfluid phase, when condensate fraction \eqref{eq:Fraction} is close to one, most of the atoms tend to occupy the zero momentum mode. In this regime, the following part of the Hamiltonian \eqref{eq:MomentumHamiltonian} can be considered, to effectively describe the system:
	\begin{equation} \label{eq:ZeroMomentumHamiltonian}
        \hat{\mathcal{H}}_{BH, q=0} = \epsilon_{0} \hat{\tilde{N}}_{0}^A + \epsilon_{0} \hat{\tilde{N}}_{0}^B + \frac{U_{AA}}{2M} \hat{\tilde{N}}_{0}^{A^2} + \frac{U_{BB}}{2M} \hat{\tilde{N}}_{0}^{B^2} + \frac{U_{AB}}{M} \hat{\tilde{N}}_{0}^{A} \hat{\tilde{N}}_{0}^{B},
    \end{equation}
    where $\hat{\tilde{N}}_{0}^{A}, \hat{\tilde{N}}_{0}^{B}$ are numbers of atoms in the zero momentum mode defined as
    \begin{subequations}
        \begin{align}
        \hat{\tilde{N}}_{0}^{A} &:= \hat{\tilde{a}}_{0}^{\dagger} \hat{\tilde{a}}_{0}, \\
        \hat{\tilde{N}}_{0}^{B} &:= \hat{\tilde{b}}_{0}^{\dagger} \hat{\tilde{b}}_{0}.
        \end{align}
    \end{subequations}
    
    Assuming that the system is in the superfluid phase with condensed fraction $f_c\approx1$, one might expect that $\braket{\hat{\tilde{N}}_{0}^{A}+\hat{\tilde{N}}_{0}^{B}} \approx N$ and therefore the Bose-Hubbard model for q=0 takes the form:
    \begin{equation*}
        \hat{\mathcal{H}}_{BH, q=0} = \frac{U_{AA}+U_{BB} - 2 U_{AB}}{2M} \hat{S}_{z_0}^2 + \frac{U_{AA}-U_{BB}}{2M} \hat{S}_{z_0} \hat{\tilde{N}}_{0} +
    \end{equation*}
    \begin{equation} \label{eq:SzHamiltonian}
        + \frac{U_{AA}+U_{BB}+4 U_{AB}}{8M} \hat{\tilde{N}}_{0}^2 + \epsilon_{0} \hat{\tilde{N}}_{0},
    \end{equation}
    where
    \begin{subequations}
        \begin{align}
            \hat{S}_{z_0} &:= \frac1{2} \left( \hat{\tilde{N}}_{0}^{A} - \hat{\tilde{N}}_{0}^{B} \right),\\
            \hat{\tilde{N}}_{0} &:= \hat{\tilde{N}}_{0}^{A} + \hat{\tilde{N}}_{0}^{B}.
        \end{align}
    \end{subequations}
    The first term in \eqref{eq:SzHamiltonian} has the form of the famous one-axis twisting model \cite{kitagawa1993squeezed}. Here the $z$ component of the collective vector refers to the zero-momentum mode only.

    It is already widely understood that the one-axis twisting model, from the initial spin-coherent state, generates squeezed and entangled states, including Schr\"odinger cat state \cite{interferReview}. The latter one is nothing else as the rotated GHZ state that this thesis focuses on. Taking this into account, one can expect that the Bose-Hubbard model generates squeezed and entangled states in the superfluid regime. Therefore, the idea of this work, as explained in Subsection~\refsec[sec:Protocol], is to generate an entangled state in the superfluid phase and then shift the system to the Mott insulator phase by rising the lattice height.

	\section{The initial state} \label{sec:initialStatePreparation}

	It is already well-recognised that squeezing and entanglement in an optical lattice can be generated from the initial spin-coherent state \cite{paper}. The spin-coherent state can be conveniently prepared in numerical simulation using the following procedure.

	One starts the simulation from the ground state $\ket{\text{GS}^A}$ of the Hamiltonian \eqref{Hamiltonian}, with additional condition of all particles being in the $A$ state. The ground state is calculated numerically. An example for $N=M=2$ is given in the Appendix~\refsec[app:Neq2].	The $\ket{\text{GS}^A}$ state is then exposed to rotation around the $y$-axis of the generalised Bloch sphere to obtain \newterm[spin-coherent state] $\ket{\text{GS}^+}$:
	\begin{equation} \label{GS+}
	\ket{\Psi(t=0)} = \ket{\text{GS}^+} := e^{-\I \frac{\pi}{2} \hat{S}_y} \ket{\text{GS}^A},
	\end{equation}
	where $\hat{S}_y$ operator is defined by \eqref{eq:SpinOperators}.

	The spin-coherent state is then evolved according to the Schr\"odinger equation:
	\begin{equation}
	i\hbar \frac{\partial}{\partial t} \ket{\Psi(t)} = \hat{\mathcal{H}}_{BH}(t) \ket{\Psi(t)}.
	\end{equation}

	In this thesis the equation is solved numerically. The equation can be solved analytically only for $N=M=2$ and constant Hamiltonian. The solution for the $N=M=2$ case is presented in Appendix~\refsec[app:Neq2].

	\section{The protocol} 
	\label{sec:Protocol}
	
	This section proposes a protocol for generating entangled states. The protocol exploits the existence of two phases in the optical lattice. It is achieved by using a model of an optical lattice filled with ultracold atoms. At the beginning, the lattice is prepared in a ground state of bosonic atoms in one internal state $A$. Then a $\pi/2$-pulse is applied to create a superposition of two internal states of atoms (it is the spin-coherent state \eqref{GS+}). The system being in superfluid phase is evolved dynamically to create entanglement. During evolution, the lattice depth is being increased to reach Mott insulator phase and eventually stop the interaction to freeze the state. The quantum correlator $\mathcal{C}_{\vec{e}}$ defined in the next section \refsec[sec:QuantumEntanglement] is used to quantify entanglement in the 6-partite system. The study finds that the appropriate choice of parameters in the protocol allows one to create maximally entangled Greenberger-Horne-Zeilinger (GHZ) state. The protocol follows the method presented in \cite{paper}. This thesis presents an analysis of entanglement in the system and introduces tools similar to those presented in \cite{Chwe20} to measure entanglement in a more general case.
	
	\begin{figure}[H]
		\centering
		\includegraphics[width=0.9\textwidth]{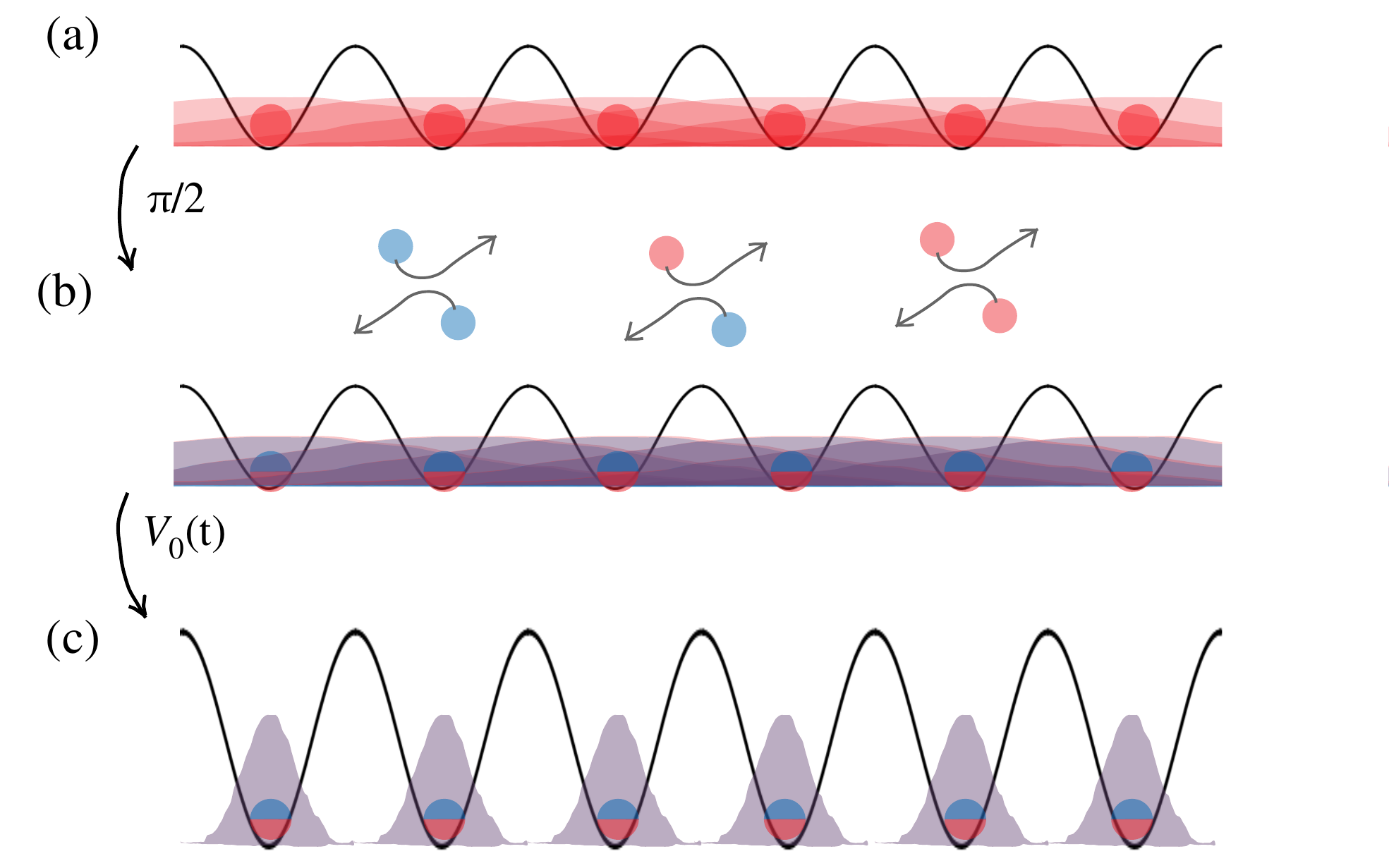}
		\caption{(a) At the beginning, in a shallow optical lattice, ultracold bosons in an internal state $A$ are prepared in the superfluid phase. The number of atoms is equal to the number of lattices. (b) Then a $\pi/2$-pulse is applied to put each atom in a coherent superposition of two internal states $A$ and $B$. As the system is in the superfluid phase, non-linear interaction between atoms starts the generation of quantum correlations in the system. (c) Depth of the lattice is being gradually increased to perform adiabatic change of the Hamiltonian and shift the system from the superfluid to the Mott-insulator phase and freeze the created correlations. Reprinted from \cite{paper}.}
		\label{optical_lattice-scheme}
	\end{figure}

	\chapter{Methods}
	
	\section{Strategy} 

	The research started from investigating the ground states of the Hamiltonian \eqref{Hamiltonian} for different parameters of the lattice. It was done to get a better understanding of the phase diagram of the system and to find the numerical values of the parameters leading to Mott insulator phase and superfluid phase (Section~\refsec[sec:phaseDiagram]).

	Then a numerical program (Appendix~\refsec[app:NumericalProgram]) was prepared to calculate the ground states of the Bose-Hubbard Hamiltonian with different values of its parameters. To construct the Hamiltonian, the Fock basis (Subsection~\refsec[sec:FockBasis]) was used. The states were represented by vectors whose components correspond to the coefficients of decomposition in the basis. The operators, such as the Bose-Hubbard Hamiltonian, were represented by matrices in the same basis. Values of condensate fraction (\ref{eq:Fraction}) and particle number variance (\ref{var_i}) for these states were calculated to specify for which range of parameters one deals with Mott or superfluid regime. The results are shown in figures \ref{phaseDiag-frac} and \ref{phaseDiag-var}.

	From the phase diagram, the appropriate parameters of the system were found to create a protocol for freezing the states. The idea was to evolve the state in the superfluid regime so it would become entangled and adiabatically shift the system into Mott regime to ``freeze'' the entanglement.

	Next, a numerical program was written to solve numerically the Schr\"odinger equation using the Runge-Kutta method \cite{RK4}. The initial state was prepared numerically according to Section~\refsec[sec:initialStatePreparation].
	During evolution the program was also calculating values of other observables such as fluctuations of particles number $\var_i$ \eqref{var_i}, superfluid fraction $f_c$ \eqref{eq:Fraction} and correlator \Cesq \eqref{eq:GeneralCorrelator}.
	A linear increase of the lattice potential was assumed what is easily achievable in an experiment. The aim was to verify whether it is possible to generate entanglement in the lattice and store it in Mott insulator phase, and possibly to find a protocol, which gives the most entangled state. Simulations for different ramp times $\tau$ were run to analyse the resulting states. To examine the entanglement in the system, Inequality (\ref{eq:Bound}) was derived and used as a criterion for separability. The value of \Cesq was employed as a measure of entanglement.

	\section{Computer tools} 
	
	To perform numerical simulations, programming languages such as \verb|Fortran 95|, \verb|C++| and \verb|Pyhton| were used. The research involved \verb|Sparse BLAS| and \verb|LAPACK| libraries in \verb|Fortran| code, \verb|OpenBLAS|, \verb|LAPACK|, \verb|TBB| and \verb|Armadillo| in \verb|C++| code and \verb|scipy.sparse| package in the \verb|Python| code to ensure efficient computations. All the simulations were prepared in the framework of GNU/Linux operating system Ubuntu. To increase the computational efficiency, the fact that the Hamiltonian matrix is sparse was exploited and special techniques for operating on sparse matrices were applied. A brief description of the computer program realising the simulations is presented in the Appendix~\refsec[app:NumericalProgram].

	\chapter{Results} \label{ch:Results}
	
    This work presents a method of generating and storing entangled states of a few bosonic atoms in an optical lattice potential described by the Bose-Hubbard model. To prove that there is entanglement in the system, a family of correlators \Cesq was introduced in Subsection~\refsec[subs:Correlator]. The value of \Cesq is a measure of entanglement in the system as shown in Subsection~\refsec[subs:Witness]. To validate that the correlator \Cesq can be used as a witness of entanglement, another important result of this thesis -- Theorem~\refsec[the:Theorem1] leading to separability bound \eqref{eq:Bound} -- was proven.   

    In the previous chapters, the Bose-Hubbard model describing one-dimensional optical lattice was introduced (Section~\refsec[sec:BoseHubbardModel]). Then it was shown that this model can be mapped into one-axis twisting model and allows producing entangled states (Section~\refsec[sec:HamiltonianIn0q]). A protocol generating entangled states, using the spin-coherent state \eqref{GS+} as the initial state, was described in Section~\refsec[sec:Protocol]. The protocol consists of the following steps: First, the spin-coherent state $\ket{GS^+}$ \eqref{GS+} is created in the lattice. The lattice is tuned so that the system remains in the superfluid regime to enable entanglement generation. Then the lattice potential height is slowly increased to adiabatically shift the system to the Mott regime where the dynamics slows down and the level of entanglement is frozen.

    The following Section~\refsec[sec:Tuning] presents how to determine the parameters of the model (number of particles per site $\nu$ and the lattice potential height $V_0$) allowing to operate the lattice from superfluid to Mott insulator regime. The form of the states produced by the protocol is discussed in Section~\refsec[sec:InternalStructure]. For some ramp time, the resulting state \eqref{eq:OurGHZ}, at the end of the evolution, is the GHZ state. In Section~\refsec[sec:CorrelatorChoice] it is shown how to find a correlator proving that this state is the maximally entangled state. At the end, Section~\refsec[sec:Breaking] demonstrates that indeed the protocol produces entangled states and specifically the maximally entangled state.

	\section{Tuning the trap} \label{sec:Tuning}
	
	The goal of this work was to create and store an entangled state. As shown in Section~\refsec[sec:HamiltonianIn0q], Bose-Hubbard Hamiltonian creates entanglement in a system in the superfluid regime. On the other hand the Mott insulator phase is necessary to freeze the state. This section presents how to choose the proper values of the number of particles per site $\nu$ and the lattice potential height $V_0$ to prepare the system in the superfluid phase and shift it to the Mott insulator phase.
	
	The Mott insulator phase is only achievable for integer number of particles per lattice site \cite{BH_model}, as it is apparent from the phase diagram of the Bose-Hubbard model (Figure~\ref{phaseDiag-var}). This study examines a system with one particle per lattice site ($N = M$). It has an additional feature that in the deep Mott regime the system can be described as a spin-$1/2$ chain \cite{spin_chain}. Due to computational limitations, it was chosen to consider a system with six lattice sites ($M=6$).

	A simulation with linear change of potential from $3 E_R$ to $45 E_R$ in a very long time (ramp time $\tau=10000~\hbar/E_R$) was run to provide adiabatic evolution. The variance of particles' number and condensate fraction was calculated to investigate for which ranges of potential there is Mott insulator or superfluid phase.
	
	\begin{figure}[H]
		\centering
		\includegraphics[width=0.8\textwidth]{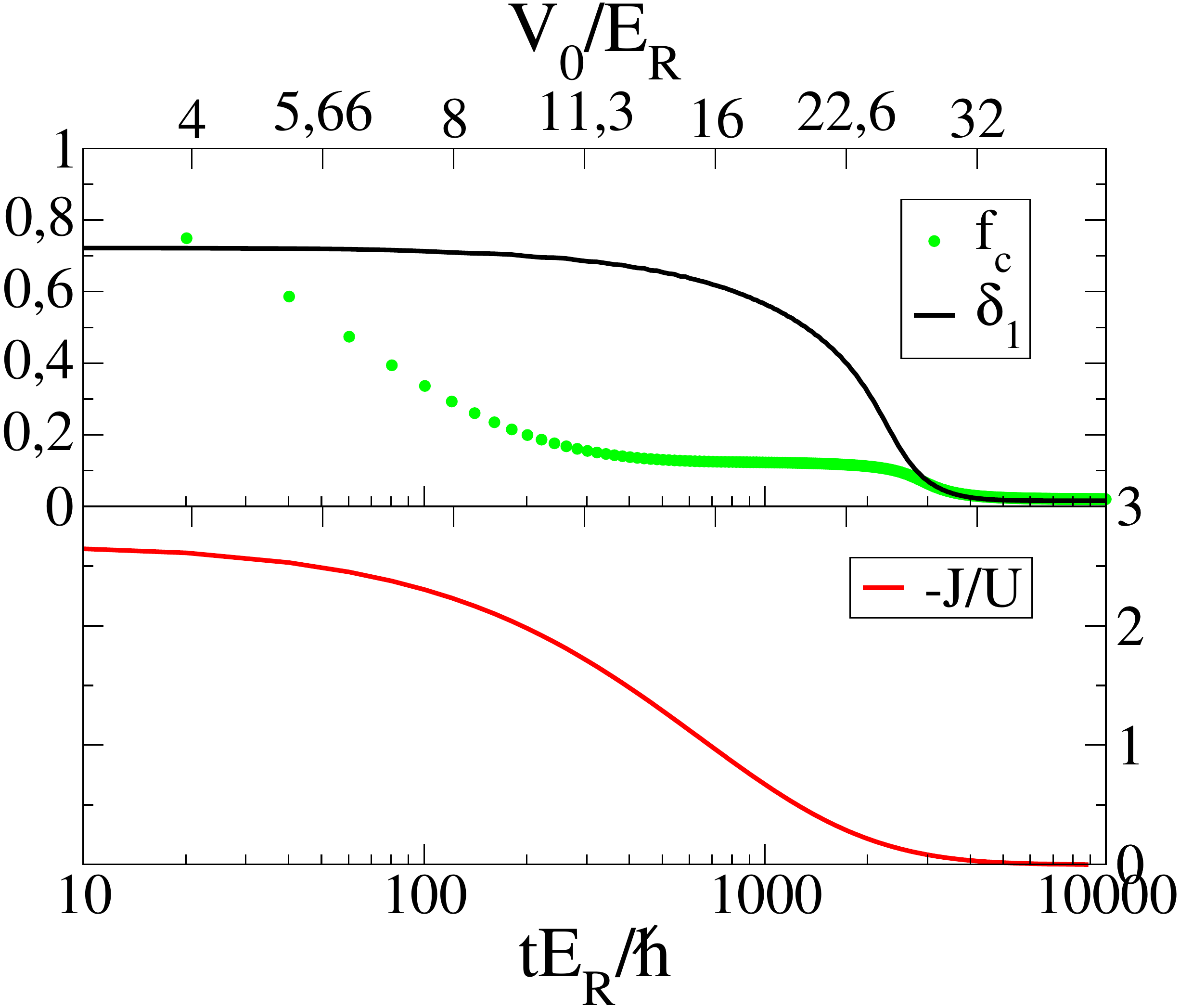}
		\caption{A plot showing the condensate fraction $f_c$ (\ref{eq:Fraction}) and the variance of the number of particles in the first site $\delta_1$ (\ref{var_i}) as a function of the lattice potential $V_0$ during adiabatic evolution of the system ($V_0$ changing slowly from $3$ to $45 E_R$).}
		\label{fig:VarianceVsV0}
	\end{figure}
	
	As observed in Figure~\ref{fig:VarianceVsV0}, the condensate fraction $f_c$ is close to one and the system is in the superfluid regime for small values of the lattice's potential ($V_0 \lesssim 4 E_R$). In this regime, the high value of condensate fraction $f_c$ allows entanglement generation (as explained in Section~\refsec[sec:HamiltonianIn0q]). When increasing $V_0$ to around $30 E_R$, the variance of number of particles in lattice site $\delta_i$ tends to zero (Figure~\ref{fig:VarianceVsV0}). It indicates that the system arrives at the Mott regime with exactly one atom per lattice site. The interaction between the atoms is hindered. The change of the system's parameters is then suppressed and one can expect to be able to ``freeze'' the resulting state.

	\section{Internal structure of the state} \label{sec:InternalStructure}
	
	To investigate the protocol (Section~\refsec[sec:Protocol]), simulations of the system with $N=M=6$ for different ramp times $\tau$ were run. In \cite{paper} it was shown that the states at the end of the ramp are of the form:
	\begin{dmath}
	\label{final_states}
	\ket{\Psi}_{\text{MI},\,N\!=\,6} = c_{60} \left( \ket{\mathcal{S}_{6,0}} + \ket{\mathcal{S}_{0,6}} \right) + c_{51} \left(\ket{\mathcal{S}_{5,1}} + \ket{\mathcal{S}_{1,5}}\right) + \\
	+ c_{42} \left(\ket{\mathcal{S}_{4,2}} + \ket{\mathcal{S}_{2,4}}\right) + c_{33}\, \ket{\mathcal{S}_{3,3}},
	\end{dmath}
	where $\ket{\mathcal{S}_{N_A,N_B}}$ (with $N_A+N_B=N=M$) are defined as
	\begin{subequations}
	\begin{align}
	\ket{\mathcal{S}_{N_A,N_B}} := \sum_{ \substack{ \{n^A_i \}:\\ \sum_i n^A_i=N_A,\\ n^A_i \leq 1 } } \sum_{ \substack{ \{n^B_i \}:\\ \sum_i n^B_i=N_B,\\ n^A_i + n^B_i \leq 1 } } \ket{\{n^A_i\}; \{n^B_i\}} \equiv \\
	\equiv \frac{1}{N_A! N_B!} \sum_{\sigma \in S_M} \prod_{k=1}^{N_A} \prod_{l=N_A+1}^{N_A+N_B} \hat{a}^\dagger_{\sigma(k)} \hat{b}^\dagger_{\sigma(l)} \vac
	\end{align}
    \end{subequations}
	and $S_M$ are all permutations of a set with $M$ elements. Example $\ket{\mathcal{S}_{N_A,N_B}}$ states using notation introduced in \eqref{eq:FockState}:
	\begin{equation}
	    \ket{\mathcal{S}_{6,0}} = \ket{1,1,1,1,1,1;0,0,0,0,0,0},
	\end{equation}
	\begin{equation*}
	    \ket{\mathcal{S}_{1,5}} = \ket{1,0,0,0,0,0;0,1,1,1,1,1} + \ket{0,1,0,0,0,0;1,0,1,1,1,1} +
	\end{equation*}
	\begin{equation}
	    + \ket{0,0,1,0,0,0;1,1,0,1,1,1} + \ket{0,0,0,1,0,0;1,1,1,0,1,1} +
	\end{equation}
	\begin{equation*}
	    + \ket{0,0,0,0,1,0;1,1,1,1,0,1} + \ket{0,0,0,0,0,1;1,1,1,1,1,0}.
	\end{equation*}
	Values of $c_{AB}$ factors in $\ket{\Psi}_{\text{MI},\,N\!=\,6}$ state can be well approximated by:
	\begin{equation}
	c_{60} = \frac{1}{2^3}, ~~ c_{51} = \frac{1}{2^3} e^{\I \phi_{51}}, ~~ c_{42} = \frac{1}{2^3} e^{\I \phi_{42}}, ~~ c_{33} = \frac{1}{2^3} e^{\I \phi_{33}},
	\end{equation}
	where $\phi_{51}, \phi_{42}$ and $\phi_{33}$ numbers are from range $[0, 2\pi[$.
	
	From the results of simulations presented in Figure~\ref{fig:StateStructure} it can be concluded that the values of $\phi_{51}, \phi_{42}$ and $\phi_{33}$ phases depend linearly on the ramp time $\tau$.
	
	\begin{figure}[H]
		\centering
		\includegraphics[width=0.85\textwidth]{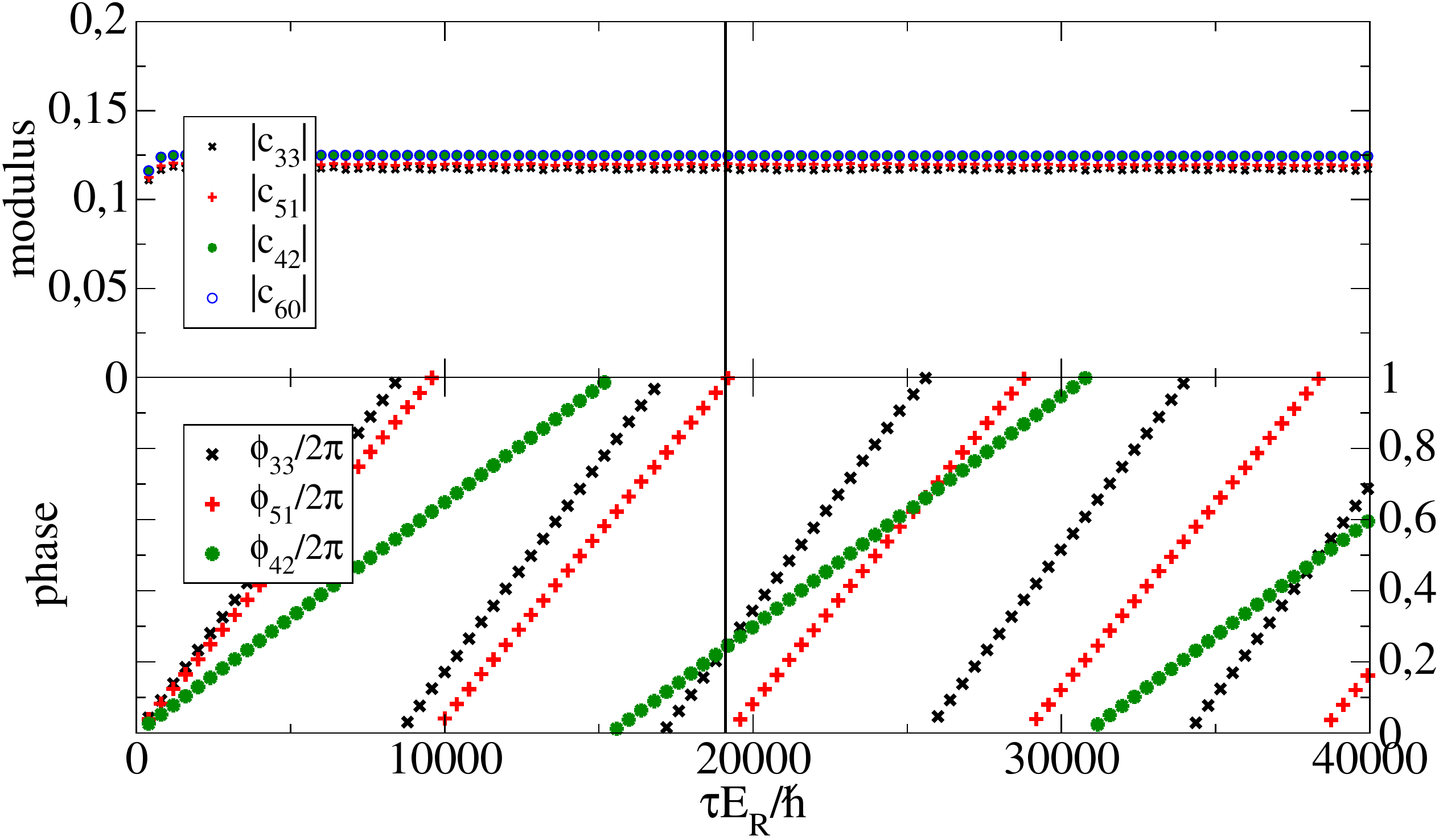}
		\caption{This figure shows internal structure of the final state for different ramp times $\tau$. The upper plot demonstrates the norm of coefficients $c_{N_AN_B}$ whose values are close to $1/2^3$ as expected. The lower plot shows phases of the coefficients, which are linearly changing with the value of $\tau$. The value $\tau=19100$ is marked with horizontal line. For this ramp time, one gets $\phi_{42}\approx0$ and $\phi_{51}\approx\frac{\pi}{2}\approx\phi_{33}$, and therefore the state is a superposition of two phase-states \eqref{eq:OurGHZ} -- which is the GHZ state. }
		\label{fig:StateStructure}
	\end{figure}
	
	To rewrite the resulting state $\ket{\Psi}_{\text{MI},N=6}$ the phase-state $\ket{\phi}_{\text{MI}}$ can be defined:
	\begin{equation}
	\ket{\phi}_{\text{MI}} := \left( \bigotimes_{j=1}^{6} \frac{\hat{a}^{{(j)}^\dagger}+e^{\I \phi}\hat{b}^{{(j)}^\dagger}}{\sqrt{2}} \right) \vac
	\end{equation}
	with $\vac$ being the vacuum Fock state. Expressing $\ket{\mathcal{S}_{N_A,N_B}}$ in terms of phase-states leads to:
	\begin{subequations}
		\label{phase_states}
		\begin{align}
		\frac{3}{4}(\ket{\mathcal{S}_{6,0}}+\ket{\mathcal{S}_{0,6}}) &= \ket{0}_{\text{MI}} + \ket{\pi}_{\text{MI}} + \ket{\frac{\pi}{3}}_{\text{MI}} + \ket{\frac{2\pi}{3}}_{\text{MI}} + \ket{\frac{4\pi}{3}}_{\text{MI}} + \ket{\frac{5\pi}{3}}_{\text{MI}}, \\	\frac{1}{2}(\ket{\mathcal{S}_{5,1}}+\ket{\mathcal{S}_{1,5}}) &= \ket{0}_{\text{MI}} - \ket{\pi}_{\text{MI}} - \I \left( \ket{\frac{\pi}{2}}_{\text{MI}} + \ket{\frac{3\pi}{2}}_{\text{MI}} \right), \\
		\frac{3}{4}(\ket{\mathcal{S}_{4,2}}+\ket{\mathcal{S}_{2,4}}) &= 2\ket{0}_{\text{MI}} + 2\ket{\pi}_{\text{MI}} - \left( \ket{\frac{\pi}{3}}_{\text{MI}} + \ket{\frac{2\pi}{3}}_{\text{MI}} + \ket{\frac{4\pi}{3}}_{\text{MI}} + \ket{\frac{5\pi}{3}}_{\text{MI}} \right), \\
		\frac{1}{2} \ket{\mathcal{S}_{3,3}} &= \ket{0}_{\text{MI}} - \ket{\pi}_{\text{MI}} + \I \left( \ket{\frac{\pi}{2}}_{\text{MI}} + \ket{\frac{3\pi}{2}}_{\text{MI}} \right).
		\end{align}
	\end{subequations}
	
	From (\ref{final_states}) and (\ref{phase_states}) it is apparent that the state at the end of the ramp $\ket{\Psi}_{\text{MI},N=6}$ defined in \eqref{final_states} can be written as a linear combination of phase-states. Moreover, for $e^{\I \phi_{42}}=1$ and $\phi_{51}=\phi_{33}=\frac{\pi}{2}$ it simplifies to the superposition of only two phase-states:
	\begin{equation*}
	\ket{\Psi} := \ket{\Psi}_{\text{MI},N=6}\left.\right|_{\phi_{42}=0, \phi_{51}=\phi_{33}=\frac{\pi}{2}} = \frac{e^{\I \pi / 4}}{\sqrt{2}} \left(\ket{0}_{\text{MI}}-\I \ket{\pi}_{\text{MI}}\right) =
	\end{equation*}
	\begin{equation}
	\label{eq:OurGHZ}
	= \frac{e^{\I \pi / 4}}{\sqrt{2}} \left( \bigotimes_{j=1}^{6} \frac{\hat{a}^{{(j)}^\dagger}+\hat{b}^{{(j)}^\dagger}}{\sqrt{2}} - \I \bigotimes_{j=1}^{6} \frac{\hat{a}^{{(j)}^\dagger}+\I\hat{b}^{{(j)}^\dagger}}{\sqrt{2}} \right) \vac.
	\end{equation}
	The above form of the state at the end of the ramp suggests that it is the GHZ state~\eqref{GHZ} (up to a phase factor between the two of superposed states). From the results of the simulation it is evident that the GHZ state can be produced when $\tau \approx 19000~E_R/\hbar$ (see Figure~\ref{fig:StateStructure}).

	\section{Choice of the proper correlator} 
	\label{sec:CorrelatorChoice}

	To prove that the state $\ket{\Psi}$ defined by Equation \eqref{eq:OurGHZ} is entangled, the correlator of the form (\ref{eq:GeneralCorrelator}) is used. The reasoning below demonstrates how to choose the direction $\vec{e}$ appearing in the definition of the correlator \eqref{eq:GeneralCorrelator} to make it useful for the system generated by the protocol. It is done for arbitrary number of sites $M$; thus, in this section the state $\ket{\Psi}$ takes more general form than \eqref{eq:OurGHZ}:
	\begin{equation}
	    \label{eq:GeneralPhaseState}
	    \ket{\Psi} = \frac{C}{\sqrt{2}} \left( \bigotimes_{j=1}^{M} \frac{\hat{a}^{{(j)}^\dagger}+\hat{b}^{{(j)}^\dagger}}{\sqrt{2}} - \I \bigotimes_{j=1}^{M} \frac{\hat{a}^{{(j)}^\dagger}+\I\hat{b}^{{(j)}^\dagger}}{\sqrt{2}} \right) \vac,
	\end{equation}
	with $C$ being some phase factor (complex number of norm one). To find a proper direction $\vec{e}$ the following reasoning can be carried out:

	\begin{quote}

		If in each of $M$ lattice sites there is exactly one atom, the lattice can be seen as a chain of $M$ qubits, where atom in state $A$ denotes $\ket{0}$ and atom in state $B$ denotes $\ket{1}$. The GHZ state will be then represented in the Fock basis by
		\begin{equation}
		\label{eq:FockGHZ}
		\ket{\text{GHZ}} := \frac{1}{\sqrt{2}} \left(\ket{0}^{\otimes M} + \ket{1}^{\otimes M}\right) \doteq \frac{1}{\sqrt{2}} \left(\bigotimes_{j=1}^M \hat{a}^{{(j)}^\dagger} + \bigotimes_{j=1}^M \hat{b}^{{(j)}^\dagger}\right)\vac.
		\end{equation}
		The analysis presented here is very general. Therefore the notation typical for description of qubits as introduced in quantum information tasks is used. Note, however, that the state $\ket{0}^{\otimes M}$ can also be expressed using the Fock state basis (see Subsection~\refsec[sec:FockBasis]), then $\ket{0}^{\otimes M} \doteq \ket{1,1,...,1;0,...,0}$.
		The correlator \eqref{eq:GeneralCorrelator} gives the largest possible value for $\vec{e}$ such, that $\hat{S}^{(j)}_{\vec{e}} = \hat{S}_+^{(j)}$ (see Subsection~\refsec[subs:Correlator]), as
		\begin{equation}
		\left| \mathcal{C}_{\vec{e}}^{\text{GHZ}} \right|^2 := \left| \braket{\bigotimes_{j=1}^M \hat{S}_{+}^{(j)}} \right|^2 = \left| \braket{\text{GHZ}|\bigotimes_{j=1}^M \hat{a}^{{(j)}^\dagger}\hat{b}^{(j)}|\text{GHZ}} \right|^2 = \frac{1}{4}
		\end{equation}
		and $1/4$ is the largest possible value of the correlator as proved in Appendix~\refsec[app:MaxCorrelator].
		
		Assuming that the state \eqref{eq:GeneralPhaseState} is the rotated maximally entangled state, one can find a rotation operator $\hat{U}^{(j)}$ acting on single qubit such that
		\begin{equation}
		    \label{eq:RotGHZ}
		    \bigotimes_{j=1}^M \hat{U}^{{(j)}^\dagger}\ket{\Psi}=\ket{\text{GHZ}}
		\end{equation}
		and then:
		\begin{equation}
		\frac{1}{4}= \left| \bra{\text{GHZ}}\bigotimes_{j=1}^M \hat{S}_{+}^{(j)}\ket{\text{GHZ}} \right|^2 = \left| \bra{\Psi}\bigotimes_{j=1}^M \hat{U}^{(j)} \hat{S}_{+}^{(j)} \hat{U}^{{(j)}^\dagger} \ket{\Psi} \right|^2.
		\end{equation}
		From the above, it is clear that one can use $\hat{S}_{rot}:=\bigotimes_{j=1}^M \hat{S}_{rot}^{(j)}:=\bigotimes_{j=1}^M \hat{U}^{(j)} \hat{S}_{+}^{(j)}\hat{U}^{{(j)}^\dagger}$ operator to prove entanglement of the $\ket{\Psi}$ state.
		
		With $\hat{U}^{(j)} \hat{U}^{{(j)}^\dagger} = \identity^{(j)}$ being ansatz it can be found that sufficient conditions for $\hat{U}^{{(j)}^\dagger}$ fulfilling \eqref{eq:RotGHZ} are
		\begin{subequations}
		    \label{eq:RotationOperator}
			\begin{align}
			\hat{U}^{{(j)}^\dagger} \hat{a}^{{(j)}^\dagger} = \frac{\hat{a}^{{(j)}^\dagger} + \I \hat{b}^{{(j)}^\dagger}}{\sqrt{2}}, \\
			\hat{U}^{{(j)}^\dagger} \hat{b}^{{(j)}^\dagger} = \frac{\hat{a}^{{(j)}^\dagger} - \I \hat{b}^{{(j)}^\dagger}}{\sqrt{2}},
			\end{align}
		\end{subequations}
	    because
		\begin{equation*}
		    \bigotimes_{j=1}^M \hat{U}^{{(j)}^\dagger} \ket{\Psi} = \bigotimes_{j=1}^M  \hat{U}^{{(j)}^\dagger} \frac{C}{\sqrt{2}} \left( \bigotimes_{j=1}^M \frac{\hat{a}^{{(j)}^\dagger}+\hat{b}^{{(j)}^\dagger}}{\sqrt{2}} - \I \bigotimes_{j=1}^M \frac{\hat{a}^{{(j)}^\dagger}+\I\hat{b}^{{(j)}^\dagger}}{\sqrt{2}} \right) \vac =
		\end{equation*}
		\begin{equation*}
		    = \frac{C}{\sqrt{2}} \left( \bigotimes_{j=1}^M \left( \hat{U}^{{(j)}^\dagger} \frac{\hat{a}^{{(j)}^\dagger}+\hat{b}^{{(j)}^\dagger}}{\sqrt{2}} \right) - \I \bigotimes_{j=1}^M  \left( \hat{U}^{{(j)}^\dagger}  \frac{\hat{a}^{{(j)}^\dagger}+\I\hat{b}^{{(j)}^\dagger}}{\sqrt{2}} \right) \right) \vac =
		\end{equation*}
		\begin{equation}
		    = \frac{C}{\sqrt{2}} \left( \bigotimes_{j=1}^M  \hat{a}^{{(j)}^\dagger} + \bigotimes_{j=1}^M \hat{b}^{{(j)}^\dagger} \right) \vac = C \ket{\text{GHZ}}.
		\end{equation}

		From \eqref{eq:RotationOperator} it is evident that:
		\begin{subequations}
			\begin{align}
			\hat{U}^{(j)} \hat{a}^{{(j)}^\dagger} &= \frac{\hat{a}^{{(j)}^\dagger} + \hat{b}^{{(j)}^\dagger}}{\sqrt{2}} \\
			\hat{U}^{(j)} \hat{b}^{{(j)}^\dagger} &= \frac{\hat{a}^{{(j)}^\dagger} -  \hat{b}^{{(j)}^\dagger}}{\I \sqrt{2}},
			\end{align}
		\end{subequations}
		and
		\begin{equation*}
		\hat{S}_{rot} := \bigotimes_{j=1}^M \hat{S}_{rot}^{(j)} := \bigotimes_{j=1}^M \hat{U}^{(j)} \hat{a}^{{(j)}^\dagger} \hat{b}^{(j)} \hat{U}^{{(j)}^\dagger} = \bigotimes_{j=1}^M \hat{U}^{(j)} \hat{a}^{{(j)}^\dagger} (\hat{U}^{(j)} \hat{b}^{{(j)}^\dagger})^\dagger =
		\end{equation*}
		\begin{equation}
		= \bigotimes_{j=1}^M \frac{\I}{2}\left(\hat{a}^{{(j)}^\dagger} \hat{a}^{(j)} - \hat{a}^{{(j)}^\dagger}\hat{b}^{(j)} + \hat{b}^{{(j)}^\dagger}\hat{a}^{(j)} - \hat{b}^{{(j)}^\dagger}\hat{b}^{(j)} \right) = \bigotimes_{j=1}^M \left( \I \hat{S}^{(j)}_z + \hat{S}^{(j)}_y \right).
		\end{equation}
		
		The $\hat{U}^{(j)}$ operator is then rotation among $z$ and $y$ axes:
		\begin{equation}
		\hat{U}^{(j)} = e^{- \I \frac{\pi}{2} \hat{S}^{(j)}_z} e^{\I \frac{\pi}{2} \hat{S}^{(j)}_y},
		\end{equation}
		what is a unitary matrix, and thus, fulfils the ansatz. One can also easily check that $\hat{S}_{rot}^{(j)}$ belongs to the family of operators $\hat{S}_{\vec{e}}^{(j)}$ \eqref{eq:SpinEOperator} with $\vec{e}=(0,1,\I)$ fulfilling Condition \eqref{eq:VecEConditions}.
	\end{quote}
	
	From the above it is apparent that measuring the value of $\left|\braket{\hat{S}_{rot}}\right|^2$ on the final state $\ket{\Psi}$ \eqref{eq:OurGHZ} one gets $\sfrac{1}{4}$ what is greater than $2^{-12}$ and proves non-separability of the state. It is worth mentioning that the $\hat{S}_{+}^{(j)}$ operator was rotated, but conversely, the state could be rotated to get the GHZ state and the $\hat{S}_+^{(j)}$ operator used to prove entanglement.
	The next section presents the results of calculations of
	\begin{equation}
	    \left| \mathcal{C} \right|^2 := \left|\braket{\hat{S}_{rot}}\right|^2 = \left|\braket{\bigotimes_{j=1}^M \left( i\hat{S}^{(j)}_z + \hat{S}^{(j)}_y \right) }\right|^2
	\end{equation}
	in numerical simulations.

	\section{Breaking the separability bound} \label{sec:Breaking}

	Following the results analysed in Section~\refsec[sec:Tuning], it was chosen to operate $V_0$ between $3 E_R$ and $40 E_R$, to ensure that at the beginning the system remains superfluid (what guarantees entanglement generation) and it ends up in the Mott insulator phase (so that evolution is frozen). The lattice height was changed linearly during time $\tau$. To check what is the optimal speed of changing the potential, simulations with different $\tau$'s were run and value of \Csq indicating entanglement in the system at the end of evolution was observed. The results are presented in Figure~\ref{t_best}.
	
	\begin{figure}[H]
		\centering
		\includegraphics[width=\textwidth]{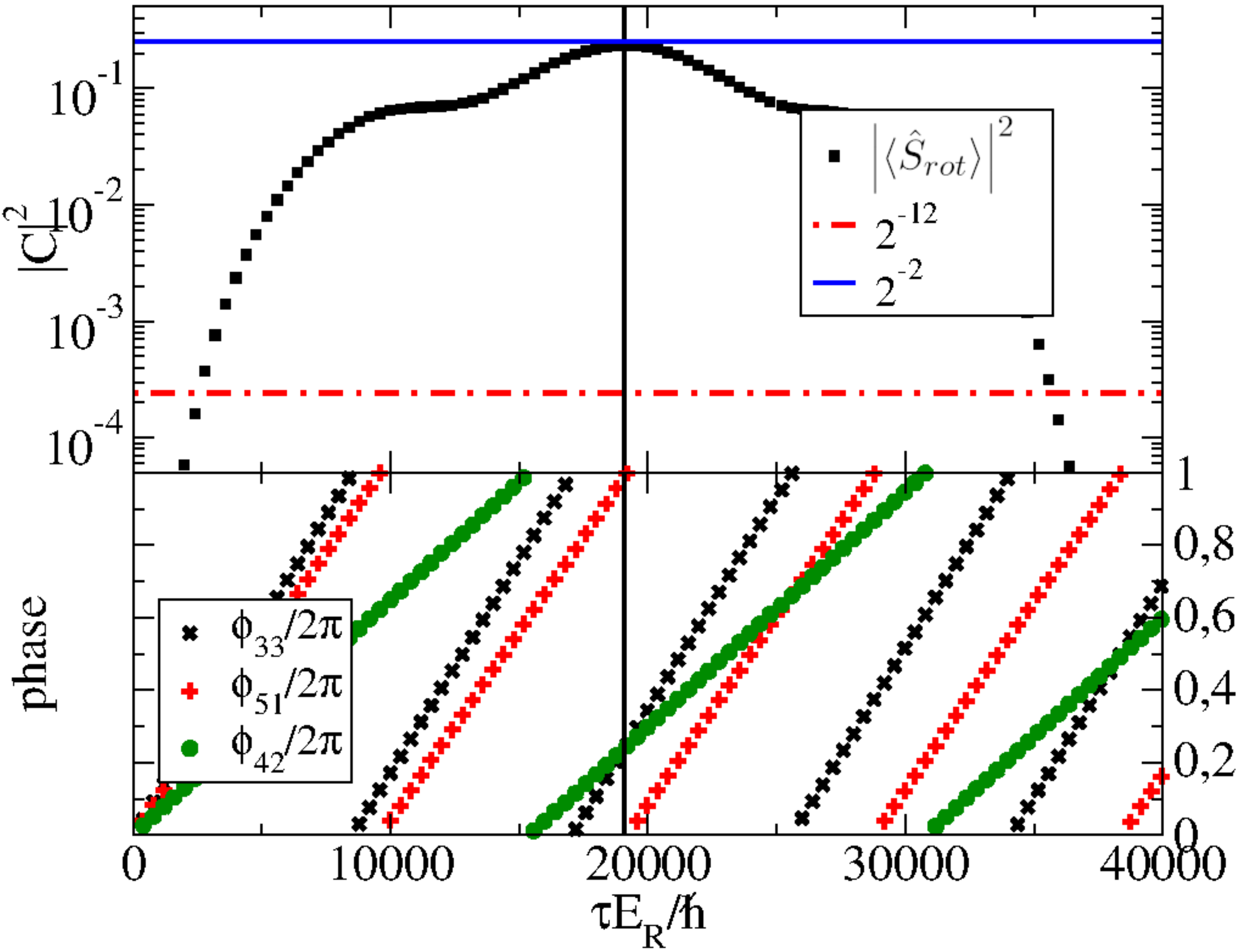}
		\caption{Each point on the plot represents a value of some variable at the end of the evolution for given $\tau$. The upper subplot shows the final values of \Csq together with the separability bound $2^{-12}$ (red dashed line) and the maximal possible value of \Csq (blue line). The lower subplot presents phases $\phi_{N_AN_B}$. The horizontal line marks $\tau=19100$ for which the value of the correlator is maximal. For that ramp time $\phi_{42}\approx0$ and $\phi_{51}\approx\frac{\pi}{2}\approx\phi_{33}$ what gives GHZ-like state (see Section~\refsec[sec:InternalStructure]).}
		\label{t_best}
	\end{figure}

	It can be observed that for ramps with ramp time $\tau \in [ 2600, 35700 ]$ the value of correlator \Csq on the final state, at the end of evolution, is greater than the separability limit $2^{-12}$. It indicates that the state cannot be described by the separable state, and therefore that there is entanglement in the system (see Subsection~\refsec[subs:Witness]). Moreover, for $\tau \approx 19000~E_R/\hbar$ the value of the correlator approaches the value $0.232$ which is close to $\sfrac{1}{4}$. It is the maximum possible value of the correlator (see Appendix~\refsec[app:MaxCorrelator]). At the same time the values of the phases of the nonzero coefficients of decomposition of the state in the Fock state basis (see Section~\refsec[sec:InternalStructure]) take the values $\phi_{42}\approx0$ and $\phi_{51}\approx\frac{\pi}{2}\approx\phi_{33}$. These values are expected when the GHZ state  \eqref{eq:OurGHZ} is realised in the system (see Section~\refsec[sec:InternalStructure]).
	
	In the following, a single evolution of the state in a ramp with ramp time $\tau=19000~E_R/\hbar$ will be considered to investigate how the entanglement does emerge in the system. The results of the simulation are presented in Figure~\ref{fig:Breaking}. Value of \Csq during the evolution is marked with the solid green line. The red dashed line represents the separability bound \eqref{eq:Bound}.
	
	\begin{figure}[H]
		\centering
		\includegraphics[width=0.9\textwidth]{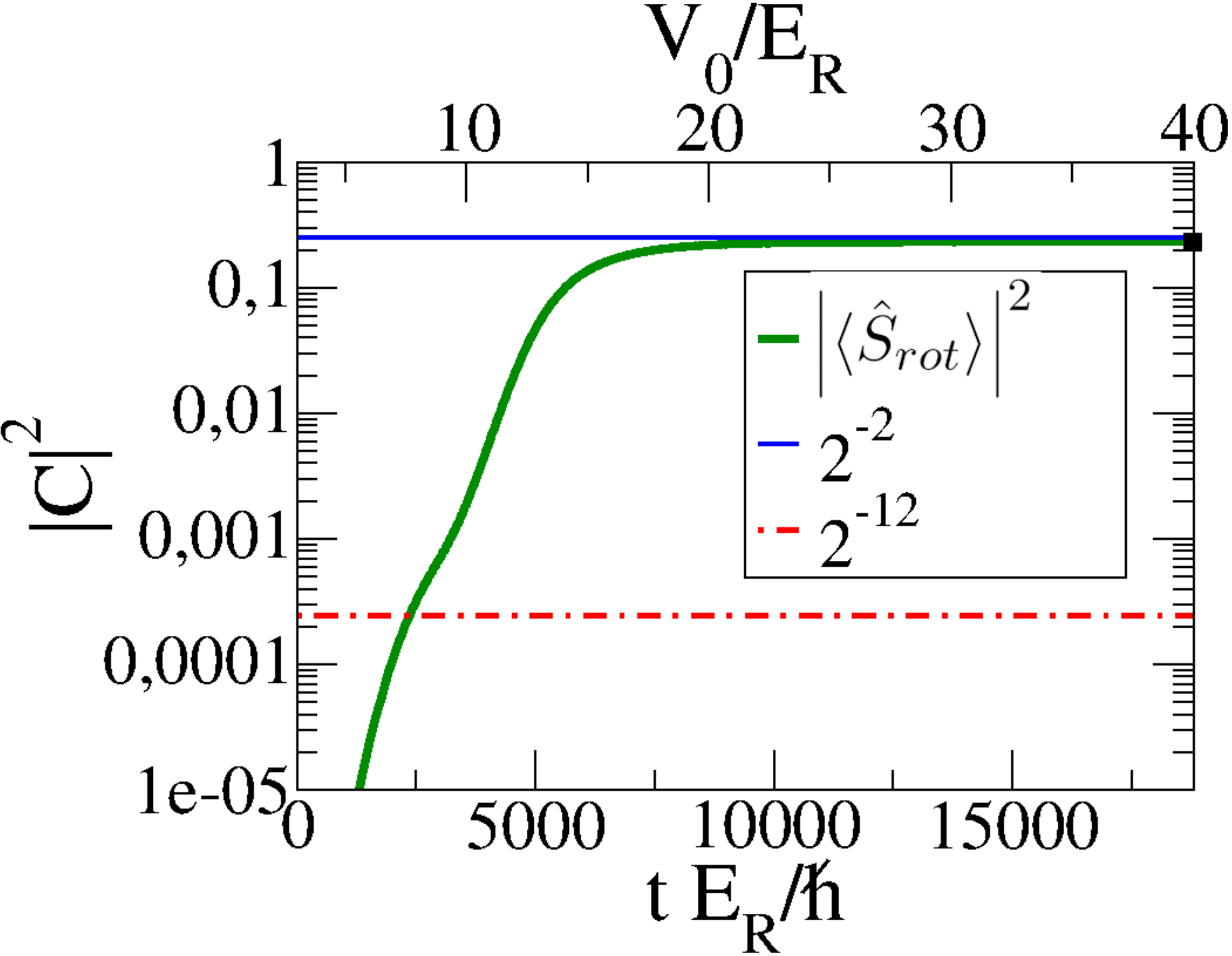}
		\caption{Horizontal axis marks the time of evolution and lattice height $V_0$ (which changes linearly with $t$). Vertical axis symbolises value of the correlator. Green line marks value of \Csq during evolution of the system. Red dotted line represents the bound for separable states for $M=6$ subsystems and blue solid line the maximal possible value of \Csq -- $1/4$. It is evident that the system breaks the separability bound and approaches highly entangled state.}
		\label{fig:Breaking}
	\end{figure}

    It is apparent that after time $2365~E_R/\hbar$ the state must be entangled as the value of \Csq is greater than the separability bound. The blue line marks $\sfrac{1}{4}$ -- the maximum possible value of the correlator. After time $10000~E_R/\hbar$ value of the correlator approaches $\sfrac{1}{4}$ and the state becomes almost maximally entangled. When the lattice depth $V_0$ exceeds $30 E_R$, the system is in the Mott insulator phase (see Figure~\ref{fig:VarianceVsV0}). Thanks to that, evolution of the state is frozen and the value of the correlator does not change anymore. Noticeably, it is a result of a numerical experiment, where decoherence processes are not taken into account.

	The above demonstrates that it is possible to use an optical lattice filled with a few bosonic atoms to create and store the maximally entangled state. To prove that, the correlator of form \Cesq \eqref{eq:GeneralCorrelator} was used. This work shows that the family of correlators \Cesq can be used to measure entanglement in an optical lattice. As \Cesq is a well-defined physical observable, it might be measured in a real experiment.

	\chapter{Summary and conclusion}

	The thesis presents a protocol for generation and storage of the maximally entangled GHZ state in the Mott phase. The method can be realised experimentally with a few ultracold atomic bosons loaded into one-dimensional optical lattice potential. It can be achieved by simply preparing all the atoms in one internal state, applying $\pi/2$-pulse before the evolution and linearly changing the lattice height.

	Numerical results for $N=6$ particles on $M=6$ lattice sites are shown. The author has also observed a similar behaviour for $N=M=4$ case. The same effects are expected to occur for an arbitrary large number of particles as long as the decoherence effects can be neglected. What is worth mentioning is that in the model used in the study, the number of particles per lattice site is exactly equal to one, change of the potential is slow enough that the system evolves adiabatically and are no particle losses. In the work, the correlator \Cesq was also defined. It measures the level of entanglement in the system and is a generalisation of the correlator proposed in \cite{Chwe20}. To generalise the applicability of the correlator from \cite{Chwe20}, analytical proofs of two inequalities: \eqref{eq:SpinInequality} and \eqref{eq:MaxCBound} are provided. Inequality \eqref{eq:SpinInequality} is used to indicate entanglement in an optical lattice filled with an arbitrary number of particles. This inequality defines a more general class of correlators that can be used to indicate entanglement in a system. Inequality \eqref{eq:MaxCBound} gives the bound for the value of the generalised correlator \Cesq. It is valid for a more general class of operators and any state with the number of atoms equal the number of lattice sites. It was proven numerically that the maximal value of the correlator \Cesq can be achieved when the GHZ state is generated and stored (Figure~\ref{fig:Breaking}) in the system composed of a few ultracold atoms loaded into the optical lattice potential by using the presented protocol (Section~\refsec[sec:Protocol]).

	The scheme can be implemented experimentally with present-day techniques and the correlator can be measured using for example quantum-gas microscopy \cite{Bakr547, Weitenberg2011}. The advantage of using a system with a few atoms is the high level of control one has, both for the preparation and for the measurement. The strength of the presented method is the possibility of scaling up the system size step by step by increasing the number of lattice sites or the spatial dimension. The discovered method has a potential to allow exploration of the boundary between the microscopic and the macroscopic world.

	\appendix

	\chapter{Proof of Inequality $\dgr[3]$}
	\label{app:CumbersomeInequality}

	This appendix contains a proof of Inequality $\dgr[3]$ \eqref{eq:BoundInequalities}. To prove the inequality, first, Lemma~\refsec[lem:Pauli] is introduced. The lemma is then used to prove Theorem~\refsec[the:Theorem1], which for $N=M$ is identical with Inequality $\dgr[3]$ \eqref{eq:BoundInequalities}.

	\begin{lemma}
	\label{lem:Pauli}
	    The sum of the modulus square of expectation value of spin operators $\hat{S}_{\vec{e}}^{(j)}$ defined in \eqref{eq:SpinEOperator} on subsystems $\hat{\rho}^{(j)}$ of any separable state $\bigotimes_{j=1}^M \hat{\rho}^{(j)}$ with $N$ atoms over $M$ lattice sites is no greater than $\sfrac{N}{2}$:
	    \begin{equation}
	        \label{eq:LemPauli}
            \sum_{j=1}^M \left| \Tr \left[ \hat{\rho}^{(j)} \hat{S}_{\vec{e}}^{(j)} \right] \right| \leq \frac{N}{2}.
	    \end{equation}
	\end{lemma}
	
	\begin{proof}
	    The left-hand side of Inequality \eqref{eq:LemPauli} can be first rewritten as follows:
	    \begin{dmath}
	        \label{eq:SingleTrace}
            \left| \Tr \left[ \hat{\rho}^{(j)} \hat{S}_{\vec{e}}^{(j)} \right] \right| = \left| \Tr \left[ \hat{\rho}^{(j)} e_x \hat{S}_x^{(j)} \right] + \Tr \left[ \hat{\rho}^{(j)} e_y \hat{S}_y^{(j)} \right] + \Tr \left[ \hat{\rho}^{(j)} e_z \hat{S}_z^{(j)} \right] \right|\,\text{=} = \left| e_x \braket{\hat{S}_x^{(j)}}_{\hat{\rho}^{(j)}} + e_y \braket{\hat{S}_y^{(j)}}_{\hat{\rho}^{(j)}} + e_z \braket{\hat{S}_z^{(j)}}_{\hat{\rho}^{(j)}} \right|\,\text{=} = \left| \vec{e} \cdot \left( \braket{\hat{S}_x^{(j)}}_{\hat{\rho}^{(j)}},  \braket{\hat{S}_y^{(j)}}_{\hat{\rho}^{(j)}}, \braket{\hat{S}_z^{(j)}}_{\hat{\rho}^{(j)}} \right)^T \right|\,\text{=} = \left| \vec{e} \cdot \overrightarrow{\braket{\hat{S}^{(j)}}}_{\hat{\rho}^{(j)}} \right|,
	    \end{dmath}
	    where $\overrightarrow{\braket{\hat{S}^{(j)}}}_{\hat{\rho}^{(j)}}:=\left( \braket{\hat{S}_x^{(j)}}_{\hat{\rho}^{(j)}},  \braket{\hat{S}_y^{(j)}}_{\hat{\rho}^{(j)}}, \braket{\hat{S}_z^{(j)}}_{\hat{\rho}^{(j)}} \right)^T$. Expression \eqref{eq:SingleTrace} can be further investigated introducing unit vectors $\vec{n}^{(j)}=(n_x^{(j)},n_y^{(j)},n_z^{(j)})^T$ parallel to $\overrightarrow{\braket{\hat{S}^{(j)}}}_{\hat{\rho}^{(j)}}$, i.e.:
	    \begin{equation}
	        n_{x,y,z}^{(j)}:= \braket{\hat{S}_{x,y,z}^{(j)}}_{\hat{\rho}^{(j)}}/\left| \overrightarrow{\braket{\hat{S}^{(j)}}}_{\hat{\rho}^{(j)}} \right|.
	    \end{equation}
	    Using this notation, Expression \eqref{eq:SingleTrace} takes the following form:
	   \begin{equation}
	        \label{eq:ESProduct}
	        \left| \vec{e} \cdot \overrightarrow{\braket{\hat{S}^{(j)}}}_{\hat{\rho}^{(j)}} \right| = \left| \overrightarrow{\braket{\hat{S}^{(j)}}}_{\hat{\rho}^{(j)}} \right| \left| \vec{e} \cdot \vec{n}^{(j)} \right|.
	   \end{equation}
	   In this work they are considered only those $\vec{e}=(e_x,e_y,e_z)$ vectors, which scalar product with any unit vector of real numbers $\vec{n}$ is no greater than one:
	   \begin{equation}
	        \label{eq:GeneralConditionE}
	        \forall_{ \substack{ \vec{n}:\\ |\vec{n}|=1 }} \left| \vec{e}\cdot \vec{n} \right| \leq 1.
	   \end{equation}
	   This condition is necessary to obtain the following boundary for the expression from \eqref{eq:ESProduct}:
	   \begin{equation}
	        \label{eq:InequalityVectors}
	         \left| \vec{e} \cdot \overrightarrow{\braket{\hat{S}^{(j)}}}_{\hat{\rho}^{(j)}} \right| \leq \left| \overrightarrow{\braket{\hat{S}^{(j)}}}_{\hat{\rho}^{(j)}} \right|.
	    \end{equation}
	    Using the above inequality, the left-hand side of Inequality \eqref{eq:LemPauli} can be then bounded in the following way:
	    \begin{equation}
	        \label{eq:SpinInequality}
	        \sum_{j=1}^M \left| \Tr \left[ \hat{\rho}^{(j)} \hat{S}_{\vec{e}}^{(j)} \right] \right| \leq \sum_{j=1}^M \left| \overrightarrow{\braket{\hat{S}^{(j)}}}_{\hat{\rho}^{(j)}} \right| = \left| \braket{\overrightarrow{\hat{S}}} \right| = \frac{N}{2}
	    \end{equation}
	    as the average value of the collective spin $\braket{\overrightarrow{\hat{S}}}$ on any state with well-defined number of particles $N$, is equal to $\sfrac{N}{2}$ (see Subsection~\refsec[sec:SpinOperators]).
	    
	\end{proof}
	
	    It is worth noting that:
        \begin{dmath}
	        \label{eq:ENProduct}
	        \left| \vec{e} \cdot \vec{n} \right| = \left| e_x n_x + e_y n_y + e_z n_z \right| \text{=}
	        = \left( |e_x|^2~ n_x^2 + |e_y|^2~ n_y^2 + |e_z|^2~ n_z^2 + 2 \left( \text{Re}~ e_x ~\text{Re}~ e_y + \text{Im}~ e_x ~\text{Im}~ e_y \right)n_x n_y + \\
	        + 2 \left( \text{Re}~ e_y ~\text{Re}~ e_z + \text{Im}~ e_y ~\text{Im}~ e_z \right)n_y n_z + 2 \left( \text{Re}~ e_z ~\text{Re}~ e_x + \text{Im}~ e_z ~\text{Im}~ e_x \right)n_z n_x \right)^{1/2}
	   \end{dmath}
	   Based on the above, an obvious consequence is that when considering such $\vec{e}=(e_x,e_y,e_z)$ vectors that fulfil the following conditions:
	    \begin{subequations}
	        \label{eq:ConditionsE}
            \begin{gather}
	        \text{Re}~ e_x ~\text{Re}~ e_y = \text{Re}~ e_y ~\text{Re}~ e_z = \text{Re}~ e_z ~\text{Re}~ e_x = 0, \\
	        \text{Im}~ e_x ~\text{Im}~ e_y = \text{Im}~ e_y ~\text{Im}~ e_z = \text{Im}~ e_z ~\text{Im}~ e_x = 0, \\
	        |e_x|, |e_y|, |e_z|~ \leq 1,
	        \end{gather}
	    \end{subequations}
	    Inequality \eqref{eq:GeneralConditionE} holds, and thus, also the inequalities \eqref{eq:InequalityVectors} and \eqref{eq:SpinInequality} hold. The class of vectors satisfying \eqref{eq:ConditionsE} is narrower than of those satisfying \eqref{eq:GeneralConditionE}. The only vectors $\vec{e}$ that fulfil \eqref{eq:ConditionsE} are of the form:
	    \begin{equation}
	        \left( \alpha, \I \beta, 0 \right),~ \left( \I \alpha, \beta, 0 \right),~ \left( \alpha, 0, \I \beta \right),~ \left( \I \alpha, 0, \beta \right),~ \left( 0, \alpha, \I \beta \right) \text{or} \left( 0, \I \alpha, \beta \right),
	    \end{equation}
	    with $\alpha, \beta \in \mathbb{R}$ and $-1 \leq \alpha, \beta \leq 1$. It is worth checking first whether Inequality \eqref{eq:Bound} is broken for $\vec{e}$ defined by \eqref{eq:ConditionsE} with $\alpha,\beta=\pm1$ before numerically searching for $\vec{e}$ that fulfils the more general Condition \eqref{eq:GeneralConditionE}.\\

	With Lemma~\refsec[lem:Pauli] proven, the main statement of this appendix can be proved:

	\begin{theorem}
	    \label{the:Theorem1}
	    For separable quantum states $\int d\lambda~ p(\lambda) \bigotimes_{j=1}^M \hat{\rho}^{(j)}(\lambda)$ of $N$ atoms on $M$ lattice sites:
        \begin{equation}
            \label{eq:Theorem1}
            \int d\lambda~ p(\lambda) \prod_{j=1}^M \left| \Tr \left[ \hat{\rho}^{(j)}(\lambda) \hat{S}_{\vec{e}}^{(j)} \right] \right|^2 \leq \left(\frac{N}{2 M}\right)^{2M}.
        \end{equation}
	\end{theorem}

	\begin{proof}
	    Inequality~\eqref{eq:Theorem1} is proved for any separable state with well-defined total number of atoms and for operators $\hat{S}_{\vec{e}}^{(j)}$ defined in Subsection~\refsec[subs:Correlator] with $\vec{e}$ fulfilling \eqref{eq:GeneralConditionE}. Because the integral in \eqref{eq:Theorem1} is just an average over probability distribution $p(\lambda)$, Theorem~\refsec[the:Theorem1] holds true if for all $\lambda$
	    \begin{equation}
	    \label{eq:SufficientSquared}
	        \prod_{j=1}^M \left| \Tr \left[ \hat{\rho}^{(j)}(\lambda) \hat{S}_{\vec{e}}^{(j)} \right] \right|^2 \leq \left(\frac{N}{2 M}\right)^{2M}.
	    \end{equation}

	    Proving Inequality \eqref{eq:SufficientSquared} is then sufficient to prove Theorem~\refsec[the:Theorem1]. Proof of Inequality \eqref{eq:SufficientSquared} starts form noting that the $2M$\nth ~root of the left-hand side of \eqref{eq:SufficientSquared} can be seen as a geometric mean. Knowing that the geometric mean is no greater than the arithmetic mean\footnote{One can prove that for $M$ non-negative numbers $\{a_i\}$ their geometric mean is bounded from above by their arithmetic mean: $\sqrt[M]{\prod_{i=1}^M a_i} \leq \frac{1}{M} \sum_{i=1}^M a_i$. As in \cite{ineq}.}, one can write:
	    \begin{equation}
	        \label{eq:MeansInequality}
	        \sqrt[M]{\prod_{j=1}^M \left| \Tr \left[ \hat{\rho}^{(j)}(\lambda) \hat{S}_{\vec{e}}^{(j)} \right] \right| } \leq \frac{1}{M} \sum_{j=1}^M \left| \Tr \left[ \hat{\rho}^{(j)}(\lambda) \hat{S}_{\vec{e}}^{(j)} \right] \right|.
	    \end{equation}
	    Raising both sides of \eqref{eq:MeansInequality} to the power of $2M$, one gets a bound for the left-hand side of \eqref{eq:SufficientSquared}:
	    \begin{equation}
	        \prod_{j=1}^M \left| \Tr \left[ \hat{\rho}^{(j)}(\lambda) \hat{S}_{\vec{e}}^{(j)} \right] \right|^2 \leq \left( \frac{1}{M} \sum_{j=1}^M \left| \Tr \left[ \hat{\rho}^{(j)}(\lambda) \hat{S}_{\vec{e}}^{(j)} \right] \right| \right)^{2M}.
	    \end{equation}
	    Now Lemma~\refsec[lem:Pauli] (which holds true for all separable states with exactly $N$ atoms and for any $\vec{e}$ fulfilling \eqref{eq:GeneralConditionE}) can be applied to write
	    \begin{equation}
	        \left( \frac{1}{M} \sum_{j=1}^M \left| \Tr \left[ \hat{\rho}^{(j)}(\lambda) \hat{S}_{\vec{e}}^{(j)} \right] \right| \right)^{2M} \leq \left( \frac{1}{M} \cdot \frac{N}{2} \right)^{2M}.
	    \end{equation}

	    The above presented reasoning proves that Inequality \eqref{eq:SufficientSquared} holds true and by this that Inequality \eqref{eq:Theorem1} also holds true.	    
	\end{proof}

    \vspace{.5cm}
	In the special case, when $N=M$ Theorem~\refsec[the:Theorem1] is identical with Inequality~$\dgr[3]$~\eqref{eq:BoundInequalities}:
	\begin{equation}
	    \int d\lambda~ p(\lambda) \prod_{j=1}^M \left| \Tr \left[ \hat{\rho}^{(j)}(\lambda) \hat{S}_{\vec{e}}^{(j)} \right] \right|^2 \leq \left(\frac{1}{2}\right)^{2M}.
	\end{equation}
	It is worth noting here that the whole proof of Lemma~\refsec[lem:Pauli], and thus, the proof of Theorem~\refsec[the:Theorem1] are based on Condition \eqref{eq:GeneralConditionE} imposed on vector $\vec{e}$. That is why \eqref{eq:GeneralConditionE} is imposed on $\vec{e}$ when defining the $\hat{S}_{\vec{e}}^{(j)}$ operator \eqref{eq:SpinEOperator} and the correlator $\mathcal{C}_{\vec{e}}$ \eqref{eq:GeneralCorrelator} in Subsection~\refsec[subs:Correlator].

	\chapter{Maximum value of \Cesq}
	\label{app:MaxCorrelator}
	
	\begin{theorem}
	    \label{the:MaxC}
	    The modulus of the quantum correlator $\mathcal{C}_{\vec{e}}$ squared  calculated on a mixed state with $N = M$ atoms on $M$ lattice sites is bounded from above by $\sfrac{1}{4}$:
	    \begin{equation}
	        \label{eq:MaxCBound}
	        |\mathcal{C}_{\vec{e}}|^2 = \left|\braket{\bigotimes_{j=1}^M \hat{S}_{\vec{e}}^{(j)}}\right|^2 \leq \frac{1}{4}.
	    \end{equation}
	\end{theorem}
	
	\begin{proof}
	    The proof starts from considering the $\mathcal{C}_{+}:=\braket{\bigotimes_{j=1}^M \hat{S}_{+}^{(j)}}$ correlator:
	    \begin{equation}
	        \mathcal{C}_{+} = \Tr \left[ \hat{\rho} \bigotimes_{j=1}^M \hat{S}_{+}^{(j)} \right] = \sum_\psi p_\psi \braket{\psi|\bigotimes_{j=1}^M \hat{a}^{{(j)}^\dagger} \hat{b}^{(j)}|\psi},
	    \end{equation}
	    with vectors $\ket{\psi}$ and probabilities $p_\psi$ constituting the density matrix of the mixed state $\hat{\rho}$ describing the whole system, i.e. $\hat{\rho} = \sum_\psi p_\psi \ketbra{\psi}{\psi}$. The $\ket{\psi}$ vectors can be represented in the Fock basis $\mathcal{F}$:
	    \begin{equation}
	        \label{eq:CorrelatorPlusInFockBasis}
	        \sum_\psi p_\psi \braket{\psi|\bigotimes_{j=1}^M \hat{a}^{{(j)}^\dagger} \hat{b}^{(j)}|\psi} = \sum_\psi p_\psi \sum_{\alpha, \alpha' \in \mathcal{F}} c_\alpha^\psi c_{\alpha'}^{\psi^*} \braket{\alpha'|\bigotimes_{j=1}^M \hat{a}^{{(j)}^\dagger} \hat{b}^{(j)}|\alpha},
	    \end{equation}
	    where $\alpha, \alpha'$ address states of the Fock basis and $c_\alpha^\psi$ are coefficients of decomposition of the $\ket{\psi}$ state in the Fock basis, i.e. $c_\alpha^\psi=\braket{\alpha|\psi}$.
	    
	    Now the state $\hat{\rho}$ giving the maximal value of $\mathcal{C}_{+}$ will be found and this maximal value will be calculated. First, one can note that the action of the operator $\bigotimes_{j=1}^M \hat{a}^{{(j)}^\dagger} \hat{b}^{(j)}$ on the state $\ket{\alpha}$ gives non-zero vector only when $\ket{\alpha}=\ket{0,...,0;1,...,1}$. This is because with $N = M$ all the other states have zero particles of $B$-type in some lattice site $j$ and the correlator acting on them returns zero. The resulting vector is
	    \begin{equation}
	        \bigotimes_{j=1}^M \hat{a}^{{(j)}^\dagger} \hat{b}^{(j)} \ket{0,...,0;1,...,1} = \ket{1,...,1;0,...,0}.
	    \end{equation}
	    It is clear now that the only non-zero terms in the expression on the left-hand side of Equation \eqref{eq:CorrelatorPlusInFockBasis} are those with $\ket{\psi}$ containing $\ket{0,...,0;1,...,1}$ and $\ket{1,...,1;0,...,0}$ states. Therefore, the state $\hat{\rho}$ giving the maximal value of $|\braket{\bigotimes_{j=1}^M \hat{S}_{+}^{(j)}}|^2$ should be a pure state of form:
	    \begin{equation}
	        \hat{\rho}^* = \ketbra{\theta,\varphi},
	    \end{equation}
	    with
	    \begin{equation}
	        \ket{\theta,\varphi} := \sin \theta \ket{0,...,0;1,...,1} + e^{\I \varphi} \cos \theta \ket{1,...,1;0,...,0},
	    \end{equation}
	    where $\theta \in \left[ 0, 2 \pi \right[$ and $\varphi \in \left[ 0, 2 \pi \right[$.
	    The value of $|\mathcal{C}_{+}|^2$ can be calculated explicitly for this state:
	    \begin{equation}
	        \left| \braket{\theta,\varphi|\bigotimes_{j=1}^M \hat{S}_{+}^{(j)}|\theta,\varphi} \right|^2 = \left| \sin \theta \braket{\theta,\varphi|1,...,1;0,...,0} \right|^2 = \sin^2 \theta~ \cos^2 \theta.
	    \end{equation}
	    This value is maximised for $\sin^2 \theta = \frac{1}{2} = \cos^2 \theta$. Therefore, the maximum value of $|\mathcal{C}_{+}|^2$ is $\sfrac{1}{4}$:
	    \begin{equation}
	        \label{eq:CplusMax}
	        \max_{\hat{\rho}} \left| \mathcal{C}_{+} \right|^2 = \max_{\hat{\rho}} \left| \Tr \left[ \hat{\rho} \bigotimes_{j=1}^M \hat{S}_{+}^{(j)} \right] \right|^2 = \max_{\hat{\theta, \varphi}} \left| \Tr \left[ \hat{\rho}^* \bigotimes_{j=1}^M \hat{S}_{+}^{(j)} \right] \right|^2 = \frac{1}{4}.
	    \end{equation}
	    This way Inequality \eqref{eq:MaxCBound} is proved for the special case $\mathcal{C}_{\vec{e}}=\mathcal{C}_{+}$.
	    In the following, this result is generalised for any $\vec{e}$. First, one can note that there exists such unitary operator $\hat{U}_{\vec{e}}$ that:
	    \begin{equation}
	        \bigotimes_{j=1}^M \hat{S}_{\vec{e}}^{(j)} = \hat{U}_{\vec{e}}^\dagger \bigotimes_{j=1}^M \hat{S}_{+}^{(j)} \hat{U}_{\vec{e}}.
	    \end{equation}
	    The value of \Cesq on some state $\hat{\rho}$ can be then expressed by $\left| \mathcal{C}_{+} \right|^2$ calculated on another state $\hat{\rho}':=\hat{U}_{\vec{e}} \hat{\rho}\, \hat{U}_{\vec{e}}^\dagger$ so that
	    \begin{equation*}
	        \max_{\hat{\rho}} \left| \mathcal{C}_{\vec{e}} \right|^2 = \max_{\hat{\rho}} \left| \Tr \left[ \hat{\rho}~ \hat{U}_{\vec{e}}^\dagger \bigotimes_{j=1}^M \hat{S}_{+}^{(j)} \hat{U}_{\vec{e}} \right] \right|^2 = \max_{\hat{\rho}} \left| \Tr \left[ \hat{U}_{\vec{e}} \hat{\rho}\, \hat{U}_{\vec{e}}^\dagger \bigotimes_{j=1}^M \hat{S}_{+}^{(j)} \right] \right|^2 =
	    \end{equation*}
	    \begin{equation}
	        = \max_{\hat{\rho}'} \left| \Tr \left[ \hat{\rho}' \bigotimes_{j=1}^M \hat{S}_{+}^{(j)} \right] \right|^2 = \frac{1}{4}
	    \end{equation}
	    according to \eqref{eq:CplusMax}. Inequality~\eqref{eq:MaxCBound} must hold true, since, as shown above, the maximum value of \Cesq is 
	    $\sfrac{1}{4}$.
	    
	\end{proof}

	\addtocontents{toc}{\protect\enlargethispage{\baselineskip}}
	\chapter{Numerical program} 
	\label{app:NumericalProgram}
		
	A single simulation consisted of two stages. In the first stage, a program written in \verb|C| was creating operators' matrices (particularly the matrix of the components of the Hamiltonian \eqref{Hamiltonian}) and computing the initial state by solving the eigenproblem of the Hamiltonian. In the second stage, another program written in \verb|Python|, was first applying $\pi/2$-pulse and then realising the numerical evolution of the state and calculating values of observables during the evolution.
	
	The input arguments of the first program were the numbers of particles and sites ($N$ and $M$) and the parameters of the lattice. First, the program generated the Fock space for given $N$ and $M$. In the created Fock space, matrices representing the Hamiltonian and other observables were obtained. Then the eigenproblem for the Hamiltonian's matrix has been solved to find the initial ground state. The matrices of the operators as well as the ground state were sent as the input for the second program.
	
	The second program, basing on the output from the first one, was rotating the initial state by $\pi/2$-pulse and evolving it. Every time step of the evolution consisted of calculation of the Hamiltonian for a given moment of time and Runge-Kutta evolution according to Schr\"odinger equation. Every few time steps, the values of interesting observables were calculated to get an insight into the evolution.

	\addtocontents{toc}{\protect\enlargethispage{\baselineskip}}
	\chapter{Evolution for $N=2$: analytical results} \label{app:Neq2}

	In this appendix a semi-analytical method used to describe the dynamics of the system with $N=M=2$ is presented. It is a result of \cite{paper}.

	The Bose-Hubbard Hamiltonian (\ref{Hamiltonian}) can be considered in the Fock state basis. This thesis uses the following notation for the elements of the Fock basis for $N=M=2$:
	\begin{align}\label{eq:Fsbasis}
	\ket{1}&=\ket{2,0;0,0}, \, 
	\ket{2}=\ket{1,1;0,0}, \,
	\ket{3}=\ket{0,2;0,0},\nonumber \\
	\ket{4}&=\ket{1,0;1,0},\,
	\ket{5}=\ket{1,0;0,1},\,
	\ket{6}=\ket{0,1;1,0},\,
	\ket{7}=\ket{0,1;0,1},\nonumber \\
	\ket{8}&=\ket{0,0;2,0},\,
	\ket{9}=\ket{0,0;1,1},\,
	\ket{10}=\ket{0,0;0,2}.
	\end{align}
	For instance, using this notation, one can write: $\hat{b}^\dagger_2 \hat{b}_1 \ket{8}=\sqrt{2} \ket{9}$.
	The matrix of the Hamiltonian takes block diagonal form in the basis: 
	$\hat{{\cal H}}_{BH}=\text{diag}\{ \hat{\mathcal{H}}_{AA}, \hat{\mathcal{H}}_{AB}, \hat{{\cal H}}_{BB} \}$,
	with representation of $\hat{{\cal H}}_{AB}$ taking the form
	\begin{equation}
	\hat{{\cal H}}_{AB} \doteq \left(
	\begin{array}{cccc}
	U_{AB} & -2 J & -2 J & 0 \\
	-2 J & 0 & 0 & -2 J  \\
	-2 J & 0 & 0 & -2 J \\
	0 & -2 J & -2 J & U_{AB} \\
	\end{array}
	\right)
	\end{equation}
	in basis $\{|4\rangle, |5\rangle, |6\rangle, |7\rangle \}$ and representations of $\hat{{\cal H}}_{AA}$, $\hat{{\cal H}}_{BB}$ in bases $\{\ket{1}, \ket{2}, \ket{3} \}$, $\{\ket{8}, \ket{9}, \ket{10} \}$ respectively, taking the form 
	\begin{equation}\label{eq:app:Hss}
	\hat{{\cal H}}_{AA} \doteq
	\left(
	\begin{array}{ccc}
	U_{\sigma\sigma} & -2 \sqrt{2} J & 0 \\
	-2 \sqrt{2} J & 0 & -2 \sqrt{2} J \\
	0 & -2 \sqrt{2} J & U_{\sigma\sigma} \\
	\end{array}
	\right) \doteq  \hat{{\cal H}}_{BB}.
	\end{equation}

    The initial spin-coherent state $\ket{\text{GS}^+}$ \eqref{GS+} is defined as the ground state with condition of all atoms being in the $A$ state, rotated by the $e^{-i \frac{\pi}{2} \hat{S}_y}$ operator:
    \begin{equation}
        \ket{\text{GS}^+} = e^{-i \frac{\pi}{2} \hat{S}_y} \ket{\text{GS}^A}.
    \end{equation}
    As all the particles in $\ket{\text{GS}^A}$ are in the $A$ state, it is enough to consider the $\hat{\mathcal{H}}_{AA}$ part of the Hamiltonian only. The eigenvalues of this operator are 
	$E_1=\Omega_{AA-}/2$, $E_2=U_{AA}$, $E_3=\Omega_{AA+}/2$,
	where $\Omega_{AA\pm}=U_{AA}\pm \sqrt{64J^2+U_{AA}^2}$ and the corresponding eigenstates
	\begin{align}
	|\Psi_1\rangle &= \frac{1}{\mathcal{N}_1}\left( |1\rangle + \frac{\Omega_{AA+}}{4\sqrt{2}J} |2\rangle +  |3\rangle\right), \\
	|\Psi_2\rangle &= \frac{1}{\sqrt{2}}\left( -|1\rangle + |3\rangle\right), \\
	|\Psi_3\rangle &= \frac{1}{\mathcal{N}_3}\left( |1\rangle + \frac{\Omega_{AA-}}{4\sqrt{2}J} |2\rangle +  |3\rangle\right),
	\end{align}
	with $\mathcal{N}_1^2=2+\Omega_{AA+}^2/(32J^2)$ and $\mathcal{N}_3^2=2+\Omega_{AA-}^2/(32J^2)$.
	The $\ket{\text{GS}^A}$ state is then
	\begin{equation}
	\ket{\text{GS}^A}=\ket{\Psi_1}=\sum_{w=1}^{10} c^{A}_w \ket{w}
	\end{equation}
	where $\vec{c}^{A}=(1,\frac{\Omega_{AA=}}{4\sqrt{2}J},1,0,0,0,0,0,0,0)^{T}/\mathcal{N}_1$, as $E_1$ is the lowest eigenvalue.
	
	The $e^{- i\hat{S}_y \pi/2}$ operator in the basis \eqref{eq:Fsbasis} takes the form
	{\footnotesize 
		\begin{equation}
		\left(
		\begin{array}{cccccccccc}
		\frac{1}{2} & 0 & 0 & -\frac{1}{\sqrt{2}} & 0 & 0 & 0 & \frac{1}{2} & 0 & 0 \\
		0 & \frac{1}{2} & 0 & 0 & -\frac{1}{2} & -\frac{1}{2} & 0 & 0 & \frac{1}{2} & 0 \\
		0 & 0 & \frac{1}{2} & 0 & 0 & 0 & -\frac{1}{\sqrt{2}} & 0 & 0 & \frac{1}{2} \\
		\frac{1}{\sqrt{2}} & 0 & 0 & 0 & 0 & 0 & 0 & -\frac{1}{\sqrt{2}} & 0 & 0 \\
		0 & \frac{1}{2} & 0 & 0 & \frac{1}{2} & -\frac{1}{2} & 0 & 0 & -\frac{1}{2} & 0 \\
		0 & \frac{1}{2} & 0 & 0 & -\frac{1}{2} & \frac{1}{2} & 0 & 0 & -\frac{1}{2} & 0 \\
		0 & 0 & \frac{1}{\sqrt{2}} & 0 & 0 & 0 & 0 & 0 & 0 & -\frac{1}{\sqrt{2}} \\
		\frac{1}{2} & 0 & 0 & \frac{1}{\sqrt{2}} & 0 & 0 & 0 & \frac{1}{2} & 0 & 0 \\
		0 & \frac{1}{2} & 0 & 0 & \frac{1}{2} & \frac{1}{2} & 0 & 0 & \frac{1}{2} & 0 \\
		0 & 0 & \frac{1}{2} & 0 & 0 & 0 & \frac{1}{\sqrt{2}} & 0 & 0 & \frac{1}{2}
		\end{array}
		\right)
		\end{equation}
	}
	and the spin-coherent state is
	\begin{equation}\label{eq:Sh}
	\ket{\Psi(0)} = \ket{\text{GS}^+} = e^{- i\hat{S}_y \pi/2} \ket{\text{GS}^A} = \sum_w c_w(0) \ket{w},
	\end{equation}
	where
	\begin{equation}
	\vec{c}_w(0)=(2\sqrt{2}, \frac{a}{2}, 2\sqrt{2},4,\frac{a}{2},\frac{a}{2},4,2\sqrt{2},\frac{a}{2},2\sqrt{2} )^{T}\frac{1}{\sqrt{8^2 + a^2}}
	\end{equation}
	and $a=\frac{\Omega_{AA+}}{J}$.
	If the Hamiltonian $\hat{{\cal H}}_{BH}$ is constant, the solution to the Schr\"odinger equation $i \hbar \partial_t \ket{\Psi(t)} = \hat{\mathcal{H}}_{BH} \ket{\Psi(t)}$ is:
	\begin{equation}\label{eq:Shro}
	\ket{\Psi(t)}=e^{-i t \hat{\mathcal{H}}_{BH}/\hbar} \ket{\Psi(0)}.
	\end{equation}
	In the general case, when the Hamiltonian depends on time, the evolution of state can be calculated numerically.
	
	As the time-independent case (\ref{eq:Shro}) can be solved analytically it is still quite cumbersome in general. However, assuming the symmetric case with $U_{AA}=U_{BB}$ the result simplifies to:
	\begin{equation}
	\ket{\Psi(t)} = \sum_w c_w(t) |w\rangle,
	\end{equation}
	where
	\begin{align}
	&c_1(t)=\frac{2Je^{-it\Omega_{AA-}/2}}{\sqrt{64J^2+U_{AA}\Omega_{AA+}}},\nonumber \\
	&c_2(t)=\frac{\sqrt{\Omega_{AA+}}e^{-it\Omega_{AA-}/2}}{2\sqrt{2\omega_{AA}}}, \nonumber \\
	&c_4(t)=\frac{2Je^{-it\Omega_{AB+}/2}
		\left[ \Omega_{AB+} - \Omega_{AA+} + e^{i t \omega_{AB}}(\Omega_{AA+} - \Omega_{AB-})\right]}{\sqrt{\omega^2_{AB}(64J^2 + \Omega_{AA+}^2)}} ,\nonumber \\
	&c_5(t)=\frac{e^{-itU_{AB}/2}}{2\sqrt{2 \omega^2_{AB} (64J^2 + U_{AA}\Omega_{AA+})}} \times \nonumber \\
	&\times
	\left[ \Omega_{AA+} \omega_{AB}{\rm cos}(t\omega_{AB}/2)
	+i (64J^2+\Omega_{AA+}U_{AB}){\rm sin}(t\omega_{AB}/2) \right] , \label{eq:c1-c5}
	\end{align}
	and $c_3(t)=c_8(t)=c_{10}(t)=c_1(t)$, $c_6(t)=c_5(t)$, $c_7(t)=c_4(t)$, $c_9(t)=c_2(t)$, with $\omega_{\sigma\sigma'} = \sqrt{64J^2 + U_{\sigma\sigma'}}$ and $\Omega_{\sigma\sigma' \pm}=U_{\sigma\sigma'} \pm \omega_{\sigma\sigma'}$. The $c_6-c_{10}$ coefficients can be expressed in terms of $c_1-c_5$ due to the symmetry in respect of exchange of the type of the particle, resulting from the equality of interaction coefficients $U_{AA}=U_{BB}$. If $U_{AA} \neq U_{BB}$ those coefficients differ and take more complex forms. It is worth mentioning, that in the limit of $J\to 0$ the coefficients $c_1, c_3, c_4, c_7, c_8, c_{10}$ tend to zero, thus only the states with exactly one atom in each lattice site contribute to the state $\ket{\Psi}$. It holds also when the non-symmetric case is considered.

	\bibliographystyle{unsrt}
	\bibliography{bibliografia}
	
\end{document}